\documentclass{ieeeaccess}
\usepackage{cite}
\usepackage{amsmath,amssymb,amsfonts,amsthm}
\usepackage{algorithm, algorithmic}
\usepackage{graphicx}
\usepackage{textcomp}
\usepackage{subcaption}
\usepackage{mathtools}  
\usepackage{bbm}
\usepackage{soul}
\usepackage{braket}
\mathtoolsset{showonlyrefs}  

\def\BibTeX{{\rm B\kern-.05em{\sc i\kern-.025em b}\kern-.08em
    T\kern-.1667em\lower.7ex\hbox{E}\kern-.125emX}}
\makeatletter
\renewcommand{\fnum@algorithm}{\fname@algorithm}
\makeatother

\newcommand{\real}{\ensuremath{\mathbb{R}}}

\newcommand{\prob}{\ensuremath{\mathbb{P}}}

\newcommand{\naturals}{\ensuremath{\mathbb{N}}}

\newcommand{\subscr}[2]{#1_{\textup{#2}}}

\newcommand{\Norm}[1]{\|#1\|}

\newtheorem{theorem}{Theorem}
\newtheorem{proposition}[theorem]{Proposition}
\newtheorem{lemma}{Lemma}

\newfloat{method}{htbp}{loa}
\floatname{method}{Method}
\makeatletter\newcommand\l@method{\@dottedtocline{1}{1.5em}{2.3em}}\makeatother

\begin{document}

\history{Received 27 June 2023, accepted 11 August 2023, date of publication 21 August 2023, date of current version 23 October 2023.}
\doi{10.1109/ACCESS.2023.3307316}
\title{Quantum Search Approaches to Sampling-Based Motion Planning}

\author{\uppercase{Paul Lathrop}\authorrefmark{1}, \IEEEmembership{Student Member, IEEE},
\uppercase{Beth Boardman}\authorrefmark{2}, and Sonia Mart\'{i}nez\authorrefmark{3},
\IEEEmembership{Fellow, IEEE}}

\address[1]{Department of Mechanical and Aerospace Engineering, University of California, San Diego, La Jolla, CA, 92092 USA, and Los Alamos National Laboratory, Los Alamos, NM 87545 USA (e-mail: pdlathrop@gmail.com)}
\address[2]{Los Alamos National Laboratory, Los Alamos, NM 87545 USA (e-mail: bboardman@lanl.gov)}
\address[3]{Department of Mechanical and Aerospace Engineering, University of California, San Diego, La Jolla, CA, 92092 USA (e-mail:soniamd@ucsd.edu)}
\tfootnote{This work was supported by Los Alamos National Laboratory and is approved
for release under LA-UR-23-23623v3. This work is licensed under a Creative Commons Attribution 4.0 License.}

\markboth
{Lathrop \headeretal: Quantum Search Approaches to Sampling-Based Motion Planning}
{Lathrop \headeretal: Quantum Search Approaches to Sampling-Based Motion Planning}

\corresp{Corresponding author: Paul Lathrop (e-mail: pdlathrop@gmail.com).}

\begin{abstract}
In this paper, we present a novel formulation of traditional
sampling-based motion planners as database-oracle structures that can
be solved via quantum search algorithms. 
We consider two complementary scenarios: for simpler sparse
environments, we formulate the Quantum Full Path Search Algorithm
(q-FPS), which creates a superposition of full random path solutions,
manipulates probability amplitudes with Quantum Amplitude
Amplification (QAA), and quantum measures a single obstacle free full
path solution. For dense unstructured environments, we formulate the Quantum Rapidly Exploring Random Tree algorithm,
q-RRT, that creates quantum superpositions of possible parent-child
connections, manipulates probability amplitudes with QAA, and quantum
measures a single reachable state, which is added to a
tree. As  performance depends on the number of oracle calls and the probability of measuring good quantum states, we quantify how these errors factor into the probabilistic completeness properties of the algorithm. We then numerically
estimate the expected number of database solutions to provide 
 an approximation of the optimal number of
 oracle calls in the algorithm.
We compare the q-RRT algorithm with a classical implementation and verify quadratic
run-time speedup in the largest connected component of a 2D dense
random lattice. We conclude by evaluating a proposed approach to limit the expected number of database solutions and thus limit the optimal number of oracle calls to a given number.
\end{abstract}


\titlepgskip=-21pt

\maketitle

\section{Introduction}

The emergence of digital electronic computing in the 1940s and 1950s brought widespread changes to virtually every area of human life. More recently, in 1980, Paul Benioff presented the quantum Turing machine~\cite{PB:80}, which outlined a simple computer using the principles of quantum mechanics to represent mixed states. The concept of quantum gates~\cite{DD:85}, which fulfill a similar function to the binary logic gates of classical computing, paved the way for the emerging field of quantum computing. Physically different from traditional computing, quantum computers leverage the quantum mechanical properties of physical matter to perform calculations simultaneously.
Quantum computation is on the horizon and awaits the development of reliable physical mediums to be used in practice~\cite{AF-NG-SB-AL:22}. Candidates for physical implementation of quantum bits (qubits) include superconducting circuits~\cite{MD-AW-JM:04} (with information storage in harmonic oscillations between energy levels of an inductor-capacitor circuit), the trapped ion quantum computer~\cite{JC-PZ:95} (with information storage in stable electronic ion states), and the semiconductor quantum dot quantum computer~\cite{DL-DD:98} (with information storage in nuclei spin states). However, the theory behind quantum computing is well established and has shown the potential to dramatically impact the solutions to many complex problems, such as in physics~\cite{IB-JD-SN:12}~\cite{MB-RB-JC-AC-JC:20} and chemistry~\cite{AK-AM-KT-MT:17} simulations, cryptography~\cite{EL-RB-YC-JK:12}~\cite{DJB-TL:17}, optimization~\cite{CCC-AG-TSH-SS:19}~\cite{FN-GC-CS-DVD-SY-BP:17}, and machine learning~\cite{IC-SC-MDL:19}.

Quantum algorithms such as Grover's
Algorithm and its generalization, Quantum Amplitude
Amplification (QAA), have a proven quadratic speedup in unstructured
database searches when compared to classical algorithms~\cite{LG:96}~\cite{GB-PH-MM-AT:02}. We believe this property
allows quantum algorithms to parallelize computationally heavy steps
in motion planning. Motivated by this, we seek to explore
how quantum algorithms and quantum speedup can be applied to sampling-based motion
planning algorithms in complex spaces with dynamic constraints.
\subsection{Literature Review} 
In this section, we provide a brief account of related works
  employing quantum computation in incidental problems in robotics,
  planning, and control theory. This is followed by a brief overview
  on sampling-based motion planning. 

With respect to motion planning, quantum algorithms have been
applied to reinforcement learning
in~\cite{DD-CC-LH:07},~\cite{DD-CC-JC-TJT:10},~\cite{DD-CC-HL-TJT:08}, and~\cite{DW-AS-RK-AK-MR:21}. Quantum methods have been shown to increase speed~\cite{DD-CC-LH:07} and robustness~\cite{DD-CC-JC-TJT:10} of state-action pair learning
algorithms in gridded environments when compared to temporal
difference epsilon-greedy and softmax choice strategies. Quantum
  reinforcement learning~\cite{DD-CC-HL-TJT:08} relies on
  encoding the state-action set as an eigen-state eigen-action set,
  with probability amplitudes characterized by quantum states in order
  to update the value function~\cite{KR-QZ-SNB-PF-JRB:21}. As
  is well known, exact reinforcement learning does not scale well to
  high-dimensional discrete state and action spaces. Even when using
  neural-network function approximations, the identification of the
  best reward functions for planning tasks in complex environments is
  an open question~\cite{LC-ZJ-LC-AK-MZ:22}. Instead, we seek to apply
  quantum computing methods to sampling-based motion planners to solve
  simpler path feasibility problems. This has the advantage to provide
  fast solutions in multi-dimensional environments with
  probabilistic completeness
  guarantees~\cite{SML:06}.

Simple robotic trajectory planning is addressed in the
work~\cite{LM:15}, which uses the Quantum Evolutionary
  Algorithm~\cite{KHH-JHK:02}, to obtain optimal trajectories
  with respect to an obstacle-distance-based objective function. A quantum genetic evolutionary
  algorithm is shown to compute trajectories in a two dimensional
  obstacle environment using
  a population-crossover-mutation workflow. This is enabled via particle swarm
  optimization (PSO); however, it is known that PSO approaches to
  motion planning suffer from a host of problems, including premature convergence, the inability to
  adapt to high dimensional search spaces (due to local optima traps and the potential to be restricted to a sub-plane of the entire search hyperplane), ambiguity in optimizer
  form (to yield both useful motion plans and solutions via PSO), and ad-hoc parameter tuning~\cite{AGG:22}.

Quantum methods have been applied to several other
  motion-planning-adjacent areas within robotics. The
work~\cite{CP-MB-HP-MH-BD:19} outlines the state of the art of
quantum computation (in terms of quantum algorithms) in robotic
science and helps frame open future research topics on
  sensing and perception, ``traditional artificial intelligence'' such
  as graph search algorithms, the integration of quantum
computers into robotic and distributed systems, and testing frameworks
for quantum computation. In particular, combinatorial graph search
  algorithms may be amenable to quantum speedup through the
  application of Grover's Algorithm, quantum annealing, or quantum
  random walks. Additionally,~\cite{CP-MB-HP-MH-BD:19} outlines
  applications of quantum algorithms to inverse kinematics and optimal
  planning problems for manipulators, by means of static optimization
  and model predictive control approaches. Here, we evaluate the
  integration of Grover's Algorithm and its extension, QAA, with
  sampling-based motion planners.  While this is unaddressed
  in~\cite{CP-MB-HP-MH-BD:19}, it aligns with the general proposed
  research agenda.  The review~\cite{NBD-FLP-APA:22} outlines the
state of the art of quantum mechanics and quantum control algorithms, addressing questions of controllability, open
  and closed loop control, and feedback control methods through the
  lens of quantum computing. The work at hand focuses on the
  computation of motion plans in obstacle environments with the help
  of quantum algorithms, rather than on the computation of feedback
  controls for quantum systems.

\IEEEpubidadjcol

The speed up of search algorithms via quantum computation has also
  received attention from other application areas; see the
textbook~\cite{RP:13}.  In
  particular, Grover's Quantum Search Algorithm has been used in~\cite{SA-AA:03} to
search a physical region, with special focus on 2D grids, with the
goal of addressing information storage constraints.  The authors
  define quantum query algorithms on predefined graphs, which could in
  theory be applied to algorithms such as the A* graph search
  algorithm~\cite{PH-NN-BR:68}. However, a proven advantage of
  sampling-based motion planners over A* approaches is that they automatically tune their
  resolution as the number of samples increases.

Quantum walks are used in~\cite{FM-AN-JR-MS:11} to find a marked
element in a discrete and finite state space. If the quantum walk is ergodic and symmetric, quadratic speedup is achieved with respect to classical Markov-chain counterparts~\cite{MS:04}. Similarly, quantum walks have been applied to search over more abstract spaces; see~\cite{ESB-JD-JGG-DZ:12}
on search engine network navigation.
Quantum walks are an extension of classical random walks, and
they require state space discretization. Instead, we seek to extend quantum speedup to tree-based planners that use randomness to find samples in continuous spaces, rather than performing motion planning over a discrete graph with random walks. This approach has been proven to efficiently solve difficult planning problems compared to methods based on discrete counterparts, and can also better handle robot dynamics.

Compared to other motion planning paradigms, sampling based motion planning avoids explicit construction of obstacle spaces in favor of performing collision checks on generated samples~\cite{SML:06}. We provide an introductory set of references, and readers are encouraged to consult the textbook~\cite{SML:06} for further reading.
In sampling based motion planning, the most
  commonly used algorithms are Probabilistic Roadmaps
  (PRM)~\cite{NA-YW:96} and the Rapidly-exploring Random Trees
  (RRT)~\cite{JB-LK-JL-TL-RM-PR:96}, both of which provide samples to grow graphs and trees
  respectively. These algorithms have been extended and modified in
  their sampling
  strategies~\cite{DH-ZS:04},~\cite{ST-MM-XT-NA:07},
  exploration~\cite{LJ-AY-SLV-TS:05},~\cite{BB-OB2:07},
  collision checking~\cite{JB-MV:02},~\cite{ZS-DH-TJ-HK-JR:05},
  speed and
  optimality~\cite{SK-EF:11},~\cite{KS-OS-DH:16},~\cite{YL-WW-YG-DW-ZF:20},
  and kinodynamic constraint
  satisfaction~\cite{JC-BD-JR-PX:88},~\cite{EF-MD-EF:02}, among
  other parameters and heuristics. An extended review of the field of
  sampling based motion planning and the relative merits and
  advantages of extending motion planning algorithms to satisfy
  certain parameters can be found at~\cite{ME-MS:14}. In this
  work, we apply quantum algorithms to basic RRTs specifically as they
  are able to find fast solutions in multi-dimensional systems, with
  no discretization required, and can account for robot dynamic
  constraints. This has made possible their widespread application in
  autonomous vehicle motion planning and complex object manipulation.
  Moving forward, the benefits of this approach can only be enhanced
  by integration with quantum computing tools. To the best of our
  knowledge, this work takes a first step in this direction.

Algorithm parallelization is related to quantum computation, as
  the heart of quantum speedup lies in the ability to perform
  simultaneous calculations on superpositions of
  states~\cite{AF-NG-SB-AL:22},~\cite{JPreskill:98}. Motion
planning algorithms have been rewritten for
multi-threading~\cite{DD-TS-JC:11}, parallel tree
creation~\cite{MStrandberg:04}, and parallel computation with
GPUs~\cite{JB-SK-EF:11}. In~\cite{DD-TS-JC:11}, the authors
  devise a message passing scheme and compare performance of several
  parallel RRT schemes, such as OR Parallel RRT, Distributed RRT, and
  Manager-Worker RRT. The work~\cite{JB-SK-EF:11} identifies the
collision checking procedure as the computationally expensive
portion of sampling-based motion planning and seeks to parallelize it.
We therefore target the collision checking procedure as the main
  candidate for quantum computing speedup. Although parallel
computation is not always a tractable solution, as with single
tree creation, path planning in dense spaces with dynamic constraints
can benefit from parallelization for quantum algorithm application.

As is detailed above, quantum search algorithms have been applied to several areas within and adjacent to robotics, such as optimization, machine learning, and estimation, but have yet to be directly applied to sampling-based motion planning algorithms, which is what we seek to accomplish here.

\subsection{Contributions}
In this work, we introduce two novel formulations of path planning
algorithms using QAA. In Quantum Full Path Search (q-FPS), we describe a quantum search over a database of
randomly generated paths from a start to a goal configuration over
sparse environments. Next, we describe a Quantum Rapidly Exploring Random Tree (q-RRT) algorithm that admits
reachable states to the tree through a quantum search of a randomly
constructed database of points.
  
The main contributions of this work are the following.

\begin{itemize}
\item Creation of a strategy for achieving path planning using
    quantum computing in sparse environments with Quantum Full Path
    Search (q-FPS);
\item Re-framing of RRTs for quantum computation with the
  algorithm Quantum RRT (q-RRT);
\item Analysis of the probabilistic completeness (PC) properties and
  derivations of key probability values of interest with respect
  to adding unreachable tree elements;
\item Characterization of oracle and measurement errors, how
  these errors affect PC, and how to ensure PC properties remain
  intact;
\item Simulations of the use of quantum algorithms for
  sampling-based motion planning and verification of quadratic
  speedup;
\item Numerical simulations regarding connectivity within 2D
  square random lattices for optimal QAA application and the
    creation of a sampling method for selecting (rather than
  estimating) the optimal number of QAA applications.
\end{itemize}

\subsection{Notation} The general notation used throughout this work
is as follows. Let $d \in \naturals$, we denote by $\real^d$ the
$d$-dimensional real vector space, and by $x\in \real^d$ a vector in
it.  We denote the Euclidean norm in $\real^d$ as $
\Norm{\cdot}_2$. Let $\mathcal{N}(y,\Sigma)$ refer to the Normal
distribution with mean $y\in \real^d$ and covariance
$\Sigma\in \real^{d\times d}$. Let $\ket{z}$ refer to the quantum
state represented by the qubit $z$. Let $\mathbb{E}$ be the
expectation operator. Let $\mathbb{U}(C)$ be the uniform distribution
over $C$, and let $\mathbb{C}$ be the space of complex numbers.
\section{Quantum Computing Basics}
In this section, we introduce quantum computing basics, how
quantum algorithms can be used to solve motion planning problems, and
an explanation of Quantum Amplitude Amplification (QAA). An extended
introduction can be found at~\cite{JPreskill:98}
and~\cite{GB-PH-MM-AT:02}. A summary of pertinent
  information from these sources is presented below.

Instead of encoding information classically in bits of either $0$ or
$1$ states, quantum computers encode information in basic units called
quantum bits or \textit{qubits}~\cite{BS:95}.  A qubit is given as the
superposition of two basis quantum states, $\ket{0}$ and
$\ket{1}$. The latter two correspond to the two physical states 0 and
1, or the classical computing states. However, a qubit $\ket{\Psi}$
can exist in a superposition of $\ket{0}$ and $\ket{1}$, of the form
$\ket{\Psi} = \alpha \ket{0}+\beta\ket{1}$, with $\alpha,\beta \in
\mathbb{C},\;|\alpha|^2+|\beta|^2=1$. We say that
$\{\ket{0},\ket{1}\}$ defines a basis of quantum states.
In this way, a qubit can be given as a weighted superposition of the
basis states, meaning it can be thought of as physically existing
simultaneously in many states at once.

Quantum states in a superposition maintain probability amplitudes
$\alpha$ and $\beta$, or the relative likelihoods of measuring a
particular state of the superposition. The measurement process of the
quantum state involves the collapse of the quantum state $\ket{\Psi}$
to a base state $\{\ket{0},\ket{1}\}$ according to the measurement
probabilities $\alpha^2$ and $\beta^2$ (also known as the Born rule~\cite{MB:26}).

Qubits are placed in superpositions using the Walsh-Hadamard
  transform, a multidimensional Fourier operator which forms the
  quantum Hadamard
  gate~\cite{AB-CB-RC-DD-NM-PS-TS-JS-HW:95}. This is a unitary
  operator mapping a quantum state to an equal superposition of all
  qubit states. Since the Hadamard gate creates the superposition, it
  is key to simultaneous computation.

Quantum algorithms use superposition as a tool to perform fast and
efficient parallel computations on superpositions of states. A unitary
transformation will act on all basis vectors of the quantum state and
can simultaneously evaluate many values of a function $f(x)$ for many
inputs $x$ in a process known as quantum
parallelism~\cite{DD-CC-HL-TJT:08}. Although the probability
amplitudes $\alpha$ and $\beta$ of the system cannot be known
explicitly~\cite{MS:05}, quantum algorithms use quantum parallelism to manipulate
the amplitudes. Planning algorithms written for quantum methods can
be thought of as fully parallelized. In this paper, we intend to use
quantum algorithms in the following general way:

\begin{enumerate}
\item Identify an oracle function (or quantum black box) to check for
  configuration feasibility or path reachability.
\item Construct a database of possible paths or points.
\item Encode the database as a qubit register (i.e. a system comprising
  multiple qubits).
\item Create a superposition across all database elements.
\item Repeatedly apply QAA to increase the probability amplitude of
  the correct database elements.
\item Measure the qubit to return a single element.
\item Check the measured answer and repeat the process.
\end{enumerate}

This succinct description on how to apply QAA to a specific
  problem is inspired by the work~\cite{RO-SM-EL:19}, which
  applies quantum algorithms to financial analysis. We will use a
Boolean oracle function to evaluate the feasibility of a path and
later, the reachability to a state. In the context of quantum
  computing, a Boolean oracle function, represents a black-box
  function that is handed inputs and produces a Boolean, or
  True/False, output~\cite{BS-FM-GM-HR-GDM:21}. They are widely
  used in quantum algorithms to study complexity and runtime
  comparisons~\cite{AA:18}. We refer to feasibility as the
connectivity of a pair of points, and provable reachability refers to
whether, given a set of dynamics and a type of controller, we can
steer the system from a state to another with a reachable
obstacle-free path. Further discussion on local planning is included
in Section~IV. The actual state and environmental parameters are not
required to be explicitly known, but the Boolean output of this oracle
is assumed to be available.

Quantum Amplitude Amplification uses a Boolean oracle function
$\mathcal{X}$ to increase the probability of measuring a good state
$\Psi$. $\Psi$ is defined in terms of being a good state if and
only if $\mathcal{X}(\Psi) = 1$. The oracle function can be described
as a Phase Oracle, and it is a unitary operator which shifts all qubit
inputs by a constant phase. The QAA operator $Q$ then performs a pair
of probability amplitude reflections based upon the output of the
oracle. This results in the probability amplitude magnification of
good states and decrease of bad states. The QAA precise definition and
mechanism of action can be found at~\cite{GB-PH-MM-AT:02}, page 56. In
what follows, the QAA operator using oracle $\mathcal{X}$ is denoted
as $Q(\mathcal{X})$.

We will take advantage of the fact that QAA can perform a quantum
search on a size-$N$ unordered database for an oracle-tagged item in
$\mathcal{O}(N^{1/2})$ oracle calls, whereas classical search
algorithms require $\mathcal{O}(N)$ calls~\cite{GB-PH-MM-AT:02}.

\section{Full Path Database Search with Quantum Amplitude Amplification}\label{section3}
In this section, we outline a first algorithm for path planning based on a direct application of QAA, with an illustration of its advantages over classical methods in a particular example.

We outline a path planning
algorithm, Quantum Full Path Search, Alg.~\ref{alg:QAA} (q-FPS), which
uses QAA to search a database $D$ of completed paths. The robot is
described by state $x\in \real^d$ which is constrained within a
compact configuration space, $C \subseteq \real^d$. Let
$\subscr{C}{free}$ denote the free space, or the space within $C$
outside of all static obstacles. The goal is for the robot to navigate
a path, in $\subscr{C}{free}$, from the initial state $x_0\in
\subscr{C}{free}$ to a goal state $\subscr{x}{G}\in
\subscr{C}{free}$. The path is denoted as an ordered set of states
$\gamma: x_0,x_1,\dots ,\subscr{x}{G}$. For the path to be considered
safe, $x_i\in C_{\text{free}},~ \forall\: i$. Continuous path curves can
also be considered.

\begin{algorithm}
\caption{$\bold{1}$ Quantum Full Path Search (q-FPS)}
\begin{algorithmic}[1] \label{alg:QAA}
\renewcommand{\algorithmicrequire}{\textbf{Input:}}
\renewcommand{\algorithmicensure}{\textbf{Output:}}
\REQUIRE $x_0,\;\subscr{x}{G}, \; n,\;$oracle function $\mathcal{X}$
\ENSURE $\gamma:\; {x_0, x_1, \dots  , \subscr{x}{G}}$
\STATE Init Database $D$
\FOR{$i = 1\text{ to }2^n$}
\STATE $D(i)\leftarrow $ random path from $x_0$ to $\subscr{x}{G}$ \label{line:alg1:database}
\ENDFOR
\STATE $m = \text{QCA}(\mathcal{X},D)$ \label{line:alg1:QCA}
\STATE Enumerate $D$ via $F:\{0,1\}^n \; \rightarrow D$ \label{line:alg1:index}
\STATE Init $n$ qubit register $\ket{z} \gets \ket{0}^{\otimes n}$
\STATE $\ket{\Psi} \gets \mathbf{W}\ket{z}$
\FOR {$i = 1\text{ to }\left\lfloor\frac{\pi}{4}\sqrt{2^n/m}\right\rfloor$} \label{line:alg1:startloop}
\STATE $\ket{\Psi} \gets Q(\mathcal{X})\ket{\Psi}$
\ENDFOR \label{line:alg1:endloop}
\STATE $\gamma \gets F($measure$(\ket{\Psi}))$ \label{line:alg1:measure}
\STATE Return $\gamma$ \label{line:alg1:return}
\end{algorithmic}
\end{algorithm}

Algorithm~\ref{alg:QAA}, the Quantum Full Path Search (q-FPS) takes
as input the initial and goal states, the desired number of quantum registers
$n$ (for database size $2^n$), and an
oracle function $\mathcal{X}$. The algorithm output is a path $\gamma\in
C_{\text{free}}$ from $x_0$ to $\subscr{x}{G}$. 

The q-FPS algorithm relies on the creation of a database of full
length path solutions on line~\ref{line:alg1:database}. In order to create a database that
is likely to contain solutions, random paths should deviate from straight line behavior. In
more complex or blocked environments, higher deviation
alongside larger database sizing $n$ can lead to a higher likelihood
of a valid solution. In Alg~\ref{alg:QAA}, on line~\ref{line:alg1:QCA}, QCA refers to the
Quantum Counting Algorithm~\cite{GB-PH-AT:98}, an extension of
Grover's algorithm and the quantum phase estimation algorithm that
estimates directly the number of solutions within the
database. Line~\ref{line:alg1:index} refers to a $1-$to$-1$ mapping
from the elements of database $D$ to states of a qubit register. It
can also be thought of as a numbering scheme. Let $\mathbf{W}$ be
the Walsh-Hadamard transform.

In the loop, from lines~\ref{line:alg1:startloop}
to~\ref{line:alg1:endloop}, we apply the QAA operator (combined with oracle $\mathcal{X}$) to the qubit multiple times to increase
the amplitude of correct database entries. The exact number of
iterations depends on the database size $2^n$ and the number of
solutions $m$ in $D$, as discussed in Section~IV. In
this application, the oracle $\mathcal{X}$ functions as a black box
indicating whether a path is obstacle collision-free. If $m$ is known, then the number of applications of $Q$ that maximize the feasible paths amplitudes is,
\begin{equation}\label{eq:imax}
    \subscr{i}{max}=\left\lfloor\frac{\pi}{4}\sqrt{2^n/m}\right\rfloor\; ;
\end{equation}
see~\cite{CF:03}. If $Q$ is
further applied, the amplitudes of correct solutions will start to
decrease, as shown in Fig.~\ref{fig:amplitudes}. Lines~\ref{line:alg1:measure} to~\ref{line:alg1:return} refer
to the process of measuring the qubit, retrieving the database
path element, and returning said path.


This method provides us with a quantum algorithm approach to motion
planning problems with a quadratic speedup over the same method using
classical search algorithms. Speedup is effected on path
collision-checking, which is the most computationally heavy portion of
path planning. 

We illustrate the algorithm and speedup on the following example. The probabilities are
known because we simulate the quantum computer on a classical device. Consider a randomized
database in a $2-$dimensional obstacle environment using a $n = 10$
register qubit corresponding to a database with $1024$ random
paths. Let there be a total of $m = 5$ obstacle free solutions within $D$
(as measured by QCA) and $Q$ will be applied to the
equal-superposition qubit $\ket{\Psi}$ a total of $i = \lfloor 11.24
\rfloor$ times, calculated using Eq.~\eqref{eq:imax}. After $11$ iterations, the total probability of
measuring one of the $5$ correct solutions is $99.86\%$ and the total
probability of measuring one of the $1019$ incorrect solutions is
$0.14\%$, as shown in Fig.~\ref{fig:amplitudes}. Classically (on a
non-quantum computer), the expected value of oracle calls to find one
of five solutions in a database of size $N = 1024$ with $m = 5$
solutions is $(N/m)/2 = 102.4 $.

\begin{figure}[h]
	\centering
	\includegraphics[width=.5\textwidth]{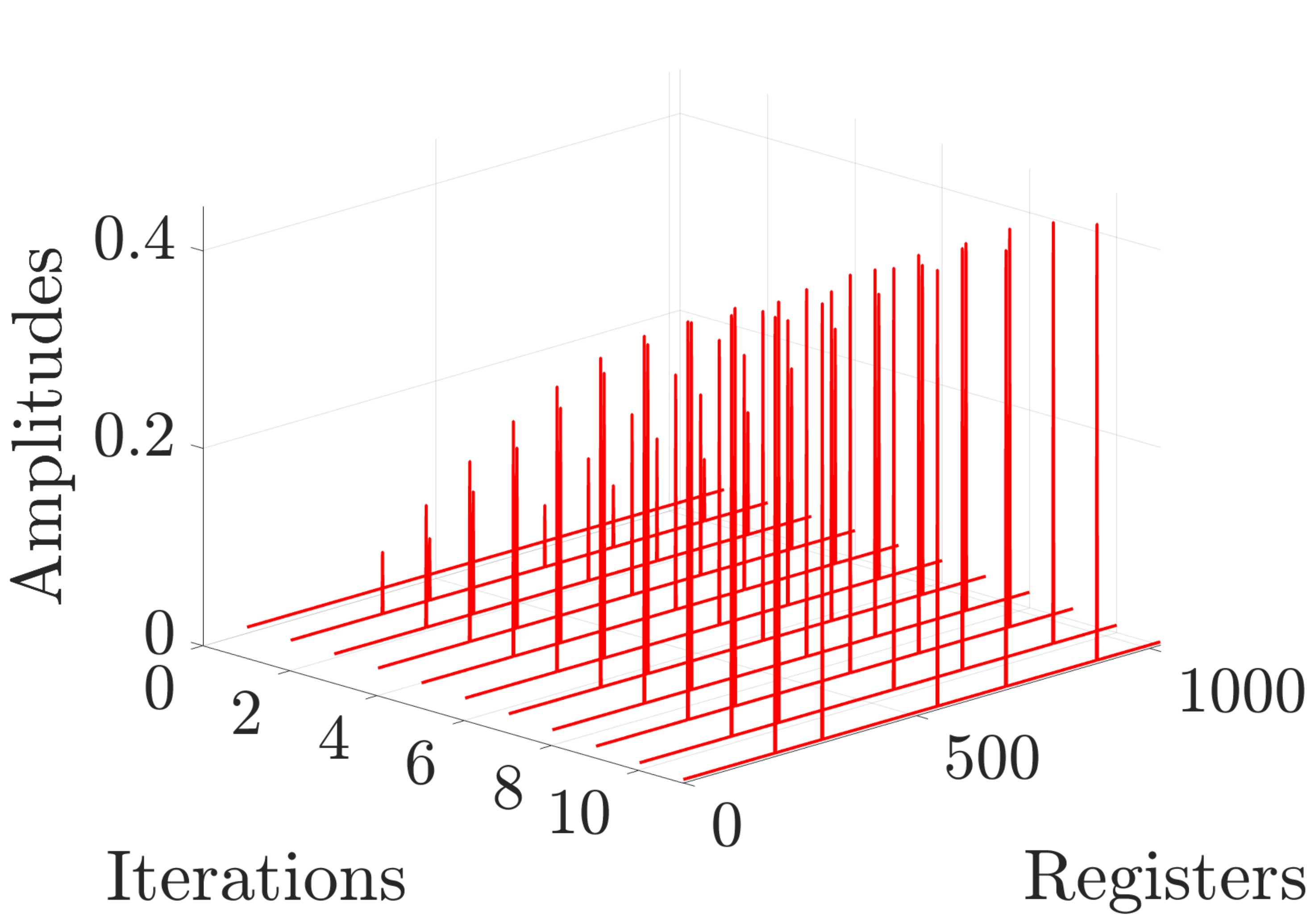}
	\caption{Effect of repeated applications of operator $Q$ on probability amplitudes of a $2^{10}$ qubit representing a database with $5$ free paths. Each register corresponds with a database element. The amplitudes of non-collision-free paths is shown as a small (non-zero) magnitude line that decreases with increased iterations. Further applications of $Q$ decrease amplitudes of free paths.}
	\label{fig:amplitudes}
\end{figure}

\section{Quantum RRT}\label{section4}
The approach of the previous section only works successfully for
obstacle-sparse environments, as randomly generated full paths are
  very unlikely to find a valid, obstacle free path when the density
  of obstacles is high. Instead, RRTs and Probabilistic
  Roadmaps (PRMs)~\cite{JJK-SML:00} are devised to produce
successful collision free-paths more quickly in cluttered
environments. In this section, we outline the q-RRT Algorithm
(Alg.~\ref{alg:q-RRT}), an RRT-like path planning algorithm, , which
is based on RRTs. The q-RRT algorithm uses QAA on a database of
individual points during tree creation to only admit reachable points
that are within the same connected component.  The main
  algorithmic differences between the q-RRT algorithm and RRT are as
  follows:
\begin{itemize}
    \item q-RRT creates databases of possible states to analyze simultaneously, rather than single states.
    \item States are assessed simultaneously for addition to the tree using quantum algorithms and measurement.
    \item A metric, $p^*$, is used to estimate the number of correct database solutions.
\end{itemize}

We analyze the algorithm performance in a $d$-dimensional finite
square (lattice) environment $C \subseteq \real^d$. The reason for
this choice is twofold: firstly, there are established tools, methods,
and theory regarding them, and secondly, they can yield sufficiently
dense and scattered environments to provide an interesting
study. Related applications include cave exploring or search and
rescue efforts in collapsed structures~\cite{BHK-MK-AME-IE:11}.

The lattice environment is shown in Fig.~\ref{fig:environment} and is defined as a square region $C =
\bigcup_{i\in \naturals} S_i \subseteq \real^d$ that is partitioned
into equal sized squares ($d=2$), cubes ($d=3$), or hypercubes ($d>3$)
$S_i,\; i \in \naturals$. Each element is either obstacle free with probability $1-r$ or occupied with probability (or concentration) $r$.
Obstacle free elements are denoted by white in our figures, and form
$\subscr{C}{free}\subseteq \real^d$, and occupied elements are denoted by black and form
$\subscr{C}{obs}\subseteq \real^d$. The characteristic length $L$ is
the ratio of the side length of $C$ to the square increment
spacing. In this section, we allow the lattice spacing to be defined
as size one and the side length of $C$ to be $L$.

Two $d$-dimensional elements are adjacent in $\real^d$ if and only if
they share a $d-1$ edge. For $d=2$, adjacency is defined for edges and
not corners. Let a connected component $Z$ be a set of
adjacent grid cells $\bigcup_{i=1} S_{i}$ such that, $
S_{i}\subseteq \subscr{C}{free},\forall i$, and any two points $x_1,x_2
\in Z$ be connected by a continuous path $\gamma\subseteq Z$.

\subsection{Quantum RRT Algorithm}

The q-RRT algorithm, Alg.~\ref{alg:q-RRT}, takes as inputs an initial point $x_0\in \subscr{C}{free}$,
the number of qubit registers $n$, a  number of nodes $M$, the
oracle function $\mathcal{X}$, a concentration $r$, and the
characteristic length $L$. It outputs a connected tree $T$ of $M$
reachable states (or tree nodes) from $x_0$. 
We note that, traditionally, RRTs end when a goal is found and return
a path. Instead, the goal is to construct an RRT that ends when
the given number of nodes $M$ are added successfully to the tree,
providing a type of PRM.

To add a node, q-RRT creates a size $2^n$ database $D$ of random
states-nearest parent pairs, as shown in
lines~\ref{line:alg2:startdata}
through~\ref{line:alg2:enddata}. The nearest parent in this
  context is defined using the $d$-dimensional Euclidean distance.
On lines~\ref{line:alg2:map} to~\ref{line:alg2:hadamard}, a $1-$to$-1$
database-element-to-qubit mapping is created and an equal
superposition is created across all qubit states. Recall that
$\mathbf{W}$ is the Walsh-Hadamard transform, the equal superposition
operator. On lines~\ref{line:alg2:Qstart}
through~\ref{line:alg2:Qend}, QAA is performed on $\ket{\Psi}$ a
repeated number of times (as per Eq.~\eqref{eq:imax}) based on an
estimate of number of solutions $m$ on Line~\ref{line:alg2:pstar},
where $p^*$ refers to estimates of $m/2^n$. A single database element
is added to the tree on line~\ref{line:alg2:treeadd} based upon the
quantum measurement on line~\ref{line:alg2:measure}. The oracle
function performs a reachability check (within the operator $Q$) with a local planner on
the random point $t$ from the proposed parent point $P$ to certify
that the returned tree is fully reachable. Our specific local planner
for simulation is explained in Section~\ref{subsection:localplanners},
and a more general discussion on reachability estimations can be found
in Sec.~\ref{section5}. We note that the method is defined as RRT (but can be extended to RRT* through the addition of standard rewiring after line~\ref{line:alg2:endwhile}) in order to apply quantum algorithms to the most broadly applicable sampling based motion planner.

\begin{algorithm}
\caption{$\bold{2}$ Quantum RRT (q-RRT)}
\begin{algorithmic}[1] \label{alg:q-RRT}
\renewcommand{\algorithmicrequire}{\textbf{Input:}}
\renewcommand{\algorithmicensure}{\textbf{Output:}}
\REQUIRE $x_0,\;n,\; M,\;\text{oracle } \mathcal{X},\;r,\;L$
\ENSURE Tree $T$
\STATE Init $T$ with root at $x_0$
\WHILE{size$(T) < M$}
\FOR{$i = 1\text{ to }2^n$} \label{line:alg2:startdata}
\STATE $t =$ random point
\STATE $P =$ closest parent of $t$ in $T$
\STATE $D(i) = [t;\:P]$
\ENDFOR \label{line:alg2:enddata}
\STATE Enumerate $D$ via $F:\{0,1\}^n \; \rightarrow D$ \label{line:alg2:map}
\STATE Init $n$ qubit register $\ket{z} \gets \ket{0}^{\otimes n}$
\STATE $\ket{\Psi} \gets \mathbf{W}\ket{z}$\label{line:alg2:hadamard}
\STATE $p^*_1 = p^*(r,L),\;p^*_2 = p^*\left(r,L/\sqrt{\text{size}(T)}\right)$ from Eq.~\eqref{eq:pstar}\label{line:alg2:pstar}
\FOR{$i = 1\text{ to }\left\lfloor\frac{\pi}{4}\sqrt{1/p^*_1}\right\rfloor$} \label{line:alg2:Qstart}
\STATE $\ket{\Psi} \gets Q(\mathcal{X})\ket{\Psi}$\label{line:alg2:Q}
\ENDFOR \label{line:alg2:Qend}
\STATE $[\subscr{x}{last},P] \gets F($measure$(\ket{\Psi}))$\label{line:alg2:measure}
\STATE Add $[\subscr{x}{last},P]$ to $T$\label{line:alg2:treeadd}
\ENDWHILE \label{line:alg2:endwhile}
\STATE Return $T$
\end{algorithmic}
\end{algorithm}
\subsection{Probabilistic Completeness and Probability Results}\label{subsection:PC}
This section analyzes the effect of two sources of error that can
affect probabilistic completeness (PC) and the admission of
unreachable states to the tree in q-RRT, leading to wrong
solutions. These are imperfect oracles and the measurement
process. The following discussion and statements apply to any path
planner with similar inaccuracies. 

In what follows, we define PC with respect to q-RRT.  For
any $x_1$ and $x_2$ that belong to the same connected component $Z\subseteq\subscr{C}{free}$, it requires that: \newline A: Eventually $x_2 \in T$, for
$T$ rooted at $x_1$ with probability $1$. \newline B: $\exists$ a
good path from $x_1$ to $x_2$ in $T$ with probability $1$. 

We relax
this standard definition to just A for the following Lemma and we
address B in Thm.~\ref{theorem:addbad}. When there are no errors, A is sufficient
because every node admitted to $T$ is reachable. We show how q-RRT can meet these criteria in Lemma~\ref{thm:PC}. 
\begin{lemma}\label{thm:PC}
For every $x_1,x_2 \in Z$, where $Z\subseteq \subscr{C}{free}$ is a connected component, the output tree $T$ of q-RRT with a final check,
with root $x_1$ satisfies $\prob(x_2\in T)\to 1$, as the number of
tested samples goes to $ \infty$.
\end{lemma}

\begin{proof}
The proof follows from the probabilistic completeness of
RRTs\cite{JJK-SML:00}. The output of RRT, $\subscr{T}{RRT}$, satisfies
$\prob(x_2\in \subscr{T}{RRT})\to 1$ as the number of samples
$\to\infty$. All points in $C$ will be tested for addition to
$\subscr{T}{q-RRT}$, similar to $\subscr{T}{RRT}$, and reachable
states will be admitted to $\subscr{T}{q-RRT}$. This result holds for
the output of q-RRT, $\subscr{T}{q-RRT}$, because the sampling
distribution (and process for selecting and admitting states) and
configuration space satisfy the same conditions as the proof for RRT,
as explained next.

In database creation, q-RRT uses independent uniform sampling of
points from within $C$, where $C$ is a nonconvex bounded open
$n$-dimensional configuration space. This distribution is multiplied
by the probability of tagging each of these states as good by the
oracle process (regardless of whether they are good or bad as ground
truth), and by the probability of measuring one of these states to be
added to the tree. It holds that $\subscr{T}{q-RRT}$ contains a tree
$\subscr{T}{RRT-m}$, which is created with only correctly identified
samples (generated by a uniform distribution over $C$) that have been
measured. The latter net distribution satisfies the necessary
conditions for the RRT result, namely that it is a smooth strictly
positive probability density function over the connected component
$Z\subseteq \subscr{C}{free}$ of interest. Then, $\subscr{T}{RRT-m}$
satisfies the theorem of RRT, and $\prob(x_2 \in \subscr{T}{RRT-m})
\rightarrow 1$. Since we have $\prob (x_2 \in \subscr{T}{RRT-m}) \le
\prob (x_2 \in \subscr{T}{q-RRT}) \le 1$, the result follows.
\end{proof}

If the oracle in Alg.~\ref{alg:q-RRT} Line~\ref{line:alg2:Q} is
imperfect, reachable states may be tagged as unreachable (false negative oracle error) and vice versa, unreachable states may be tagged as reachable (false positive oracle error), as shown in Fig.~\ref{fig:falseposfalseneg}. An ``imperfect oracle'' is one that
admits any type of error. False negative errors reduce efficiency and have the potential to remove PC
properties, as good states may not be added to the tree.
False positive errors serve to increase the
likelihood that unreachable states are admitted to the tree. 
The local planner
employed does not make repeatable false negative errors, as
reachability is defined with respect to a current state, and as the
current state approaches the target state (as discussed later), if the
target state is reachable, the oracle will identify it as
such. Therefore, oracle false negative errors do not affect PC properties.

These errors are compounded with those introduced by the measurement
step on Alg.~\ref{alg:q-RRT} Line~\ref{line:alg2:measure}, which may
admit unreachable states to the tree (additional false positive measurement error), but because the measurement produces a reachable output (and not a tag like the oracle), additional false negative measurement error is not possible.

We analyze these error-measurement likelihoods next, and their impact on property~B.  First, we note that
the false positive measurement error can be mitigated through a
final deterministic oracle check before a state is added to $T$.  We call this the ``final check'', to be applied after Alg.~\ref{alg:q-RRT} Line~\ref{line:alg2:measure}, to verify that the measured node is indeed reachable with an obstacle-free path before it is added to $T$, allowing us to use the PC definition according to solely criteria~A. However, this final check comes at a cost of additional oracle calls. 

Measurement error stems from the probabilistic nature of the qubit measurement
process (Alg.~\ref{alg:q-RRT} Line~\ref{line:alg2:measure}). In
general, there is a nonzero probability that a database element marked
(by the oracle) as bad is selected for addition to $T$ (false positive measurement error). The quantum measurement process takes a
qubit and returns a deterministic state, where the returned state
probability of selection is the square of the probability amplitude
(Born's rule)~\cite{WHZ:05}. In general, the probability amplitude of
bad states after successive applications of $Q$ is nonzero,
and the following theorem provides a characterization of this probability and its impact on criterion~B.
  
\begin{theorem}\label{theorem:addbad}
Let $E$ be the event of a bad state, as tagged by the oracle
(regardless of ground truth), being added to $T$ on a particular qubit
measurement (false positive measurement error). Let database sampling be uniform over $C$ and let the database be optimally amplified.
The probability of $E$ is,
    \begin{equation}\label{eq:ThmPE}
        \prob(E)=
        1-\sin^2\left(\left(\frac{\pi}{2}\sqrt{\frac{2^n}{m}}+1\right)\arcsin\left(\sqrt{\frac{m}{2^n}}\right)\right),
    \end{equation}
where $2^n$ is the current database size and $m$ is the current number
of solutions within the database. Eq.~\eqref{eq:ThmPE} is the minimum value that is achieved
when $Q$ is applied exactly according to
Eq.~\eqref{eq:imaxexact}\footnote{Functionally, Eq.~\eqref{eq:ThmPE}
will be modified by the fact that $Q$ is applied an integer number of
times.}.
As the number of nodes $M\to\infty$, $\prob(E)$ monotonically increases to $\lim_{M\to\infty}\prob(E)\equiv\prob(\subscr{E}{lim})$,
\begin{equation}\label{eq:ThmPEupper} 
    \prob(\subscr{E}{lim})= 
    1-\sin^2\left(\left(\frac{\pi}{2}\sqrt{1/r}+1\right) \arcsin\left(\sqrt{r}\right)\right),
\end{equation}
where $r$ is the environment concentration.
Lastly, let $F$ be the
event that at least one bad state exists within $T$. When $M$ nodes
are in $T$, an upper bound on the probability of $F$ is,
\begin{equation}\label{eq:ThmAtLeastOnce}
    \prob(F)\leq
    1-(1- \prob(\subscr{E}{lim}))^M,
\end{equation}
and an upper bound on the
probability that at least one bad state is part of a given
path~$\gamma$,
\begin{equation}\label{eq:ThmAtLeastOnceGamma}
    \prob(F_\gamma) \leq
    1-(1- \prob(\subscr{E}{lim}))^{|\gamma|},
\end{equation}
where $|\gamma|$ is the number of nodes in $\gamma$.
\end{theorem}
We remark that there is no way of finding lower bounds similar to Eq.~\eqref{eq:ThmAtLeastOnce} and Eq.~\eqref{eq:ThmAtLeastOnceGamma}, as the expected lower bound value of Eq.~\eqref{eq:ThmPE} depends on the local planner. In this case, Eq.~\eqref{eq:ThmAtLeastOnce} and Eq.~\eqref{eq:ThmAtLeastOnceGamma} form expected worst-case estimates to tree and path errors, respectively, when using q-RRT.
\begin{proof}
    First, we note that the optimal number of applications of $Q$ to
    maximize the chance of a good measurement is, 
    \begin{equation}\label{eq:imaxexact}
        \subscr{i}{max}=\frac{\pi}{4}\sqrt{\frac{2^n}{m}}\;,
    \end{equation}
    as given in~\cite{GB-PH-MM-AT:02}. We further note that, after
    $\subscr{i}{max}$ iterations, the probability of measuring a good
    state is,
    \begin{equation}\label{eq:Ec}
        \prob(E^c) = \sin^2((2\subscr{i}{max}+1)\theta),
    \end{equation}
    where $\theta$ is defined such that $\sin^2(\theta) =
    \frac{m}{2^n}$~\cite{GB-PH-MM-AT:02}, and where $\frac{m}{2^n}$ is
    the success probability of the database.
    Thm.~\ref{theorem:addbad} Eq.~\eqref{eq:ThmPE} follows via
    substitution. 
    For local
    planners testing reachability, as $M\to\infty$, in the maximal
    case the entirety of $\subscr{C}{free}$ becomes locally reachable.
    Therefore, the ratio of correct database solutions $2^n/m$
    approaches the environment concentration $r$, yielding
    Thm.~\ref{theorem:addbad} Eq.~\eqref{eq:ThmPEupper}. 
    
    Lastly, we observe that $\prob(E)$ is upper bounded by
    Eq.~\eqref{eq:ThmPEupper} and $\prob(E)$ is strictly increasing as
    a function of $m$, over our entire effective domain of $m/2^n= (0.04\;\;0.75)$. If we
    assume the upper bound for each node in the tree,
    Thm.~\ref{theorem:addbad} Eq.~\eqref{eq:ThmAtLeastOnce} follows by
    substituting Eq.~\eqref{eq:ThmPEupper} into the probability formula of at
    least one $\prob(E)$ occurring over $M$
    events. Eq.~\eqref{eq:ThmAtLeastOnce} is an upper bound over the
    number of nodes $M$, as it is found by assuming an upper bound
    value occurs in every case. 
    This is modified by replacing the
    power $M$ for $|\gamma|$ for the case of a path $\gamma$ with node
    length $|\gamma|$, yielding Eq.~\eqref{eq:ThmAtLeastOnceGamma}.
\end{proof}

We note that for databases with less than
$75\%$ solutions, $0<m<0.75*2^n$, $\prob(E)$ is strictly increasing as the fraction of
solutions in the database $m/2^n$ increases. It is also well
approximated by a linear function, $\hat{\prob}(E) = 1.251
\frac{m}{2^n} - 0.0159$, achieved with linear least squares on $m/2^n
= (0.04\;\;0.75)$ with coefficient of determination $R^2>0.999$. With
a local planner testing reachability, in general, as $M\uparrow$, the
number of database solutions $m\uparrow$. We defer this discussion to
Section~\ref{subsection:localplanners}. 

The above quantum measurement error analysis is modified in Prop.~\ref{prop:fpfnerror} to additionally account for false positive and
false negative errors by the oracle.

\begin{proposition}\label{prop:fpfnerror}
Let $G$ be the event that a good state, with respect to ground truth
(rather than as tagged by the oracle), is measured for addition to the
tree. Let the probability of a state marked incorrectly as good be
given by $q\in[0,1]$ (false positive), and let the probability of a state marked incorrectly as
bad be given by $v\in[0,1]$ (false negative). Let the database be optimally amplified. Then, the event $G$ has probability,
\begin{equation}\label{eq:probEbar}
    \prob(G) =
    \left(-1+q+v\right)\prob(E)+1-q,
\end{equation}
where $\prob(E)$ is given by Eq.~\eqref{eq:ThmPE}.
As the number of nodes $M\to \infty$, the probability
of event $G$, denoted as $\lim_{M\to\infty}\prob(G)\equiv\prob(\subscr{G}{lim})$, is given by,
\begin{equation}\label{eq:limG}
     \prob(\subscr{G}{lim})=
    \left(-1+q+v\right)\prob(\subscr{E}{lim})+1-q,
\end{equation}
where $\prob(\subscr{E}{lim})$ is given by Eq.~\eqref{eq:ThmPEupper}
and where the maximum value is again achieved when the database is
optimally amplified. Let $F^*$ be the event that at least one bad
state exists within $T$, when oracle errors are considered. When $M$
nodes are in $T$, an upper bound on the probability of $F^*$ is,
\begin{equation}\label{eq:Fstar1}
    \prob(F^*) \leq 1-( \prob(\subscr{G}{lim}))^M,
\end{equation}
and an upper bound on the
probability that at least one bad state is part of a given path
$\gamma$ is given by,
\begin{equation}\label{eq:Fstar2}
    \prob(F^*_\gamma)\leq 1-(\prob(\subscr{G}{lim}))^{|\gamma|},
\end{equation}
where $|\gamma|$ is the number of nodes in $\gamma$.
\end{proposition}
\begin{proof}
    The proof stems from modifications made to the statement of
    Eq.~\eqref{eq:ThmPE} to move from
    measurement probability with respect to the oracle to measurement
    probability with respect to ground truth. To factor in both types
    of error, $\prob(E^c) q$ (fraction of false positive error, as given in Eq.~\eqref{eq:Ec}) must be
    added to $\prob(E)$ from Eq.~\eqref{eq:ThmPE}, and $\prob(E) v$ (fraction of false negative error) must be
    subtracted from $\prob(E)$, as shown in Fig.~\ref{fig:falseposfalseneg}. This yields,
    \begin{equation}\label{eq:propproof}
        \prob(\bar{G}) = \prob(E)+\prob(E^c)q-\prob(E)v,
    \end{equation}
    where $\bar{G}$ denotes the complement of event $G$. Eq.~\eqref{eq:propproof} can be simplified, and the complement taken, to give Eq.~\eqref{eq:probEbar}. Eq.~\eqref{eq:limG} is found by taking Eq.~\eqref{eq:probEbar} and substituting $\prob(G)$ and $\prob(E)$ with $\prob(\subscr{G}{lim})$ and $\prob(\subscr{E}{lim})$, respectively. Eq.~\eqref{eq:Fstar1} and Eq.~\eqref{eq:Fstar2} are found with the same process as Eq.~\eqref{eq:ThmAtLeastOnce} and Eq.~\eqref{eq:ThmAtLeastOnceGamma} with the complement of event $G$.
\end{proof}

\begin{figure}[h]
	\centering
	\includegraphics[width=.5\textwidth]{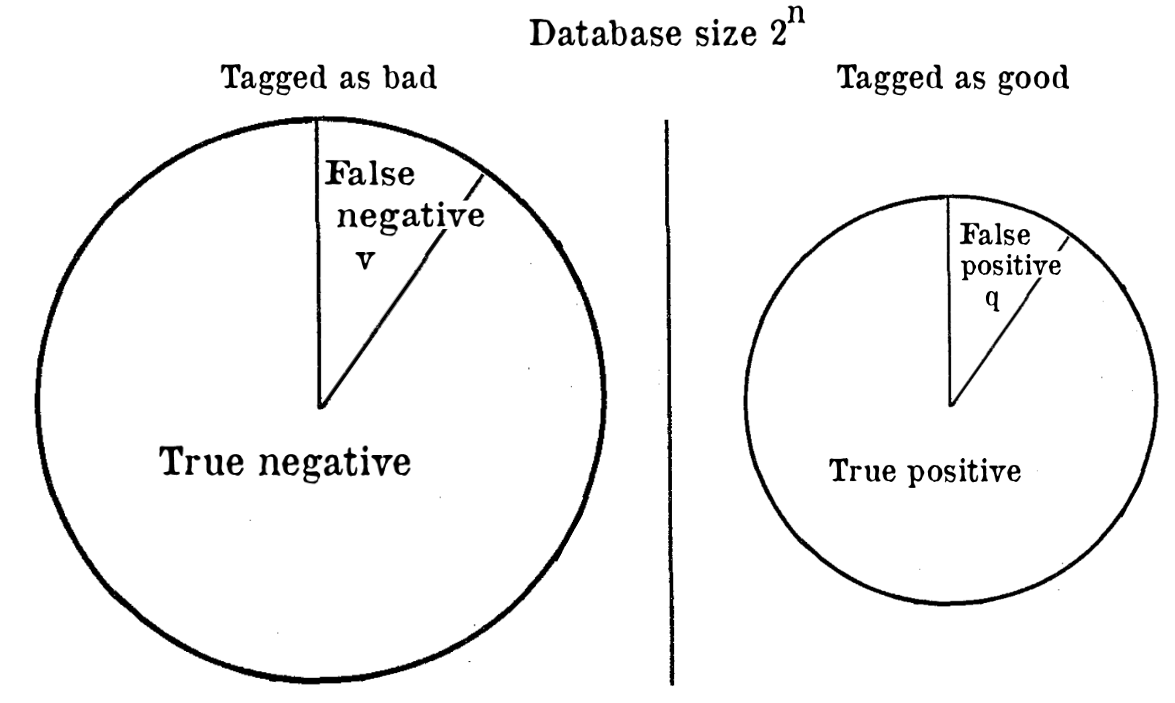}
	\caption{A visual depiction of the false positive and false negative regions of the good and bad tags by an oracle.}
	\label{fig:falseposfalseneg}
\end{figure}

A tree with as many good states as possible is achieved with the lower bounds of error in Thm.~\ref{theorem:addbad}. Attaining this bound requires applying QAA an optimal number of times, which is what we estimate next.

\subsection{Estimate of the Number of Correct Solutions}\label{subsection:estimate}

 In this section, we explore methods for estimating the number of
 tree-admittable states out of a database of uniformly random points
 inside of a $2$-dimensional periodic finite square lattice of
 characteristic length $L$ and concentration $r$. This estimation will
 guide the algorithm in applying QAA the optimal number of times (Eq.~\eqref{eq:imax}). Let the function,
\begin{equation}
    p(x_1,x_2) = \begin{cases} 1, & \text{if } x_1,x_2 \in
       Z, \\ 0, &
      \text{otherwise},
    \end{cases}
\end{equation}
represent connectivity, for a connected component $Z\subseteq \subscr{C}{free}$. Initially, we are concerned with whether or
not the two states are within the same connected component. In Section~\ref{subsection:localplanners}, we discuss local planners and reachability. We estimate the average connectivity $p^*$ of
$2$ random points within the square lattice as an estimator of the
number of correct solutions to the database $D$,
\begin{equation}\label{eq:connecprob}
    p^* = \mathbb{E}_{\pi(x_1,x_2)}(p(x_1,x_2)),
\end{equation}
where $\pi(x_1,x_2) = \mathbb{U}(\subscr{C}{free})\times \mathbb{U}(C)$.

Several results from Percolation Theory provide insight as to average
connectivity of finite square
lattices~\cite{KC:02},~\cite{RZ:92},~\cite{NJ:99}. The
work~\cite{SM-RZ:16} uses results from~\cite{MFS-JWE:64} to calculate
and estimate wrapping probabilities of $2$D square lattices. Wrapping
probabilities refer to the probability that there exists a giant
connected component from one edge of the $2$D square lattice to the
opposite edge. In the context of q-RRT, since each parent is assumed
to be in $\subscr{C}{free}$, wrapping probabilities, as presented
in~\cite{SM-RZ:16}, cannot be directly used. Additionally, our desired
estimation is with respect to individual points and not a set of points representing an
edge, as in~\cite{SM-RZ:16}.

\begin{figure}[h]
	\centering
	\includegraphics[width=.5\textwidth]{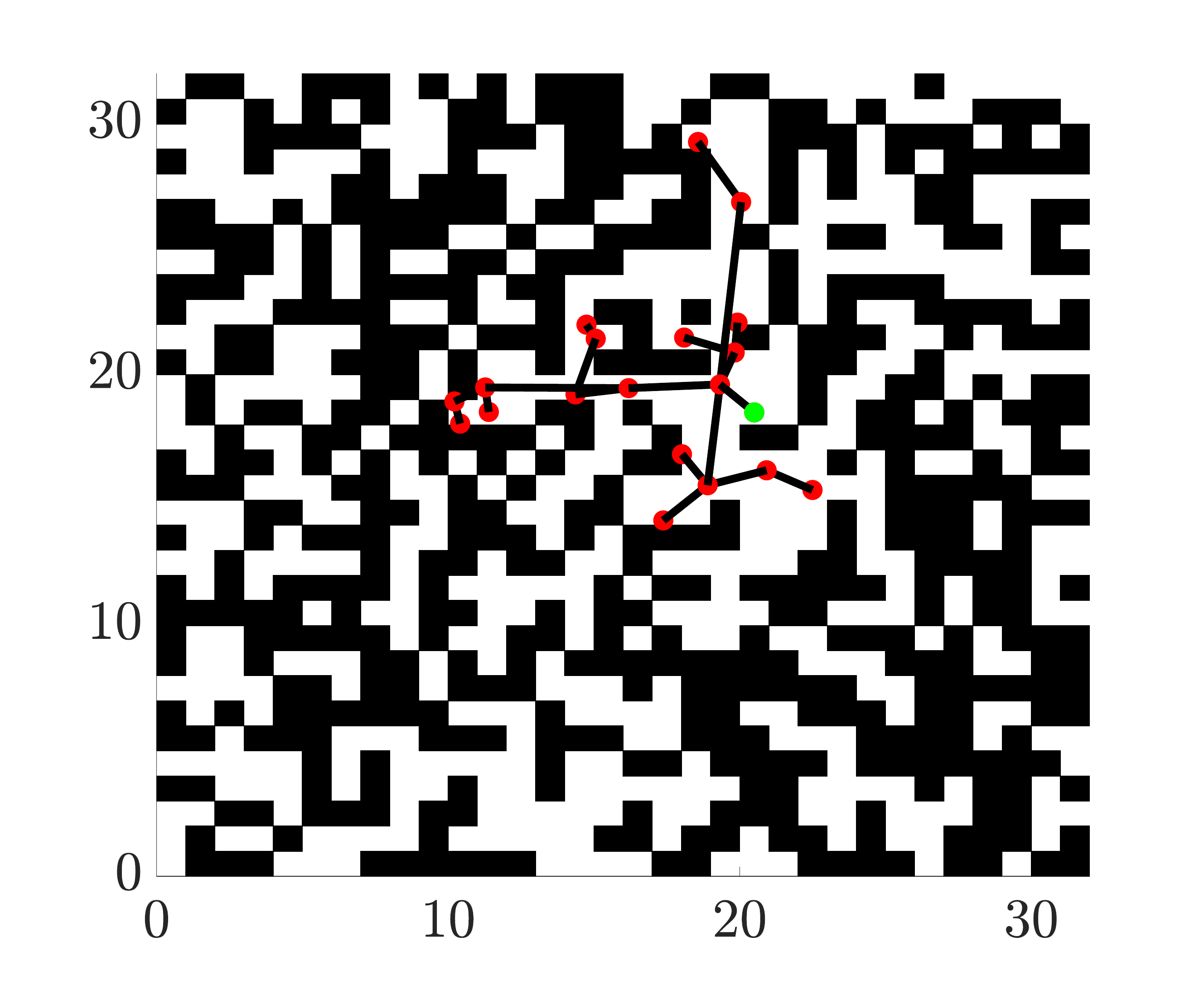}
	\caption{A sample random square lattice with $L = 32$ and $r = 0.5$ spanned by a $20$ node tree with $x_0$ in green.}
	\label{fig:environment}
\end{figure}

We calculate the connection probability, Eq.~\eqref{eq:connecprob},
from a state $x_1\in \subscr{C}{free}$ to a random state $x_2\in C$. This
reflects an estimate for correct solutions to the database in the case
where all nodes of the tree reduce to the root $x_0$. In the next
section, we evaluate the case where trees that are maximally spread in
the environment. We fit a model to numerical simulations over
concentration $r$ and characteristic length $L$ to estimate connection
probability $p^*$,

\begin{equation}\label{eq:pstar}
    p^*(r,L) = \frac{f}{1+e^{-a(L-b)(r-c)}}+dL^{-2},
\end{equation}
with $a = -0.1597, \; b = -54.59,\;c = 0.3212,\;d = 1.195,\;f =
0.9542$, found with nonlinear least squares and with an
ordinary coefficient of determination $R^2 = 0.9957$;
see~\cite{NJDN:91}. The model was chosen as a logistic
function due to observations on matching function data in~\cite[Fig.~5]{SM-RZ:16}. While $(p^*,\;r)$ slices of data exhibit a
logistic relation, it is not independent of $L$ based on inspection of
level sets in $L$, which is therefore modeled as a scaling parameter
of the logistic function. It is observed that $(p^*,\;L)$ slices
exhibit a negative nonlinear relation which is modeled with a
quadratic.

Each point $\circ$ in Fig.~\ref{fig:infillsigmoid} in the parameter
space $(r,L)$ represents the average of $1,000$ random connectivity
tests over $25$ different random lattices each, totaling $25,000$
points. In aggregate, data was collected over $209$ parameter-space
points, totaling $5.225$ million data points. The total dataset was
condensed, and the model was trained on averages because we seek to
estimate averages. Since the number of data points ($209$ averages) is
large compared to the number of parameters ($5$), we are not concerned
with over-fitting and therefore report the coefficient of
determination $R^2$ and do not split the data into training and
validation sets. We refer the reader to
  Section~\ref{sec:validation-a} for an evaluation of this metric.

\begin{figure}[h]
	\centering
	\includegraphics[width=.5\textwidth]{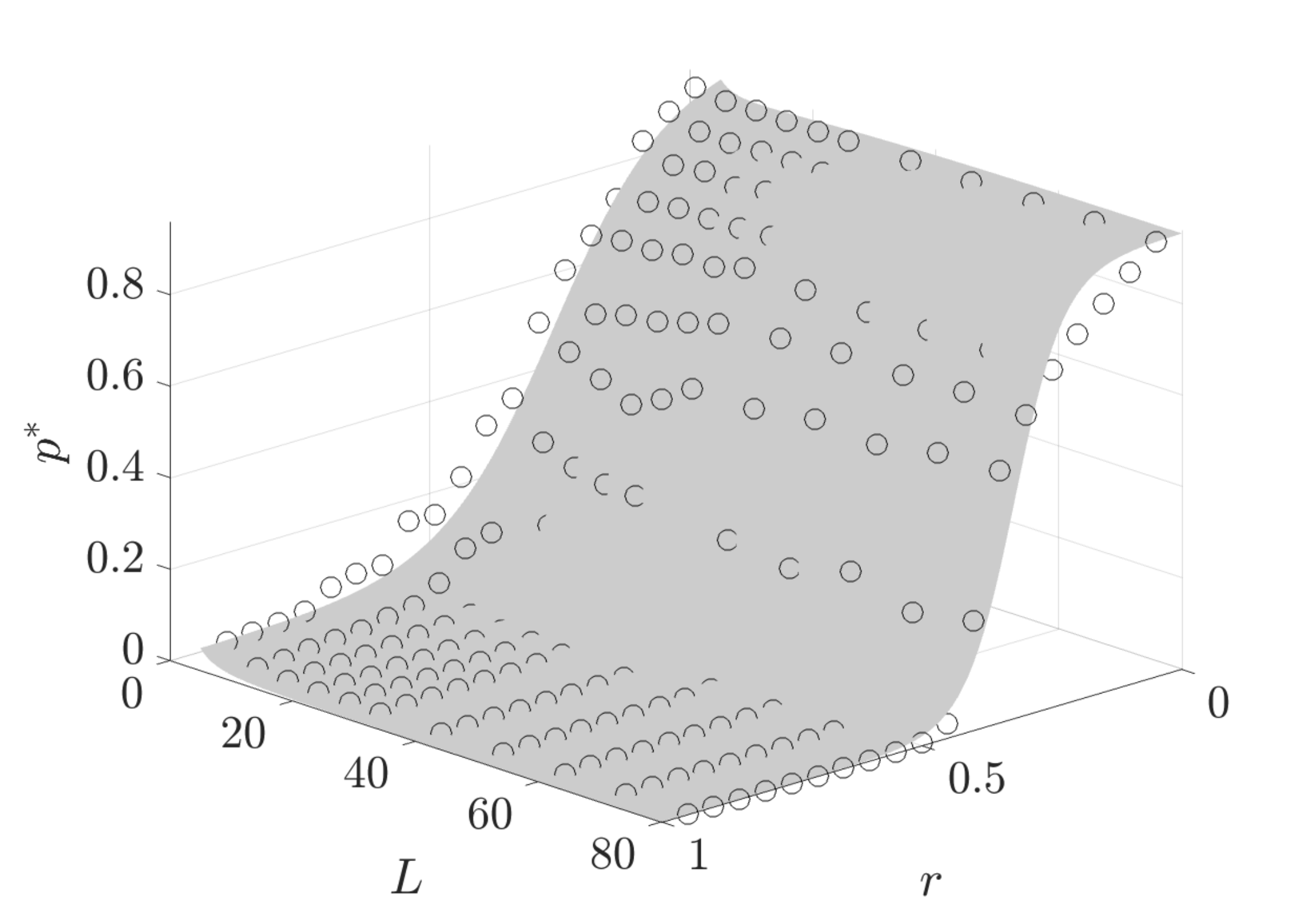}
	\caption{Numerically generated data points ($\circ$)
          estimating $p^*$ (free-random point connectivity) as a
          function of concentration $r$ and length
          $L$. Eq.~\eqref{eq:pstar} is depicted as the gray surface,
          with a coefficient of determination $R^2 = 0.9957$.}
	\label{fig:infillsigmoid}
\end{figure}
\subsection{Local Planners and Upper Bound Limit}\label{subsection:localplanners}
The choice of local planner affects the accuracy and, therefore, the
relevance of Section~\ref{subsection:estimate}. Previously, we sought to
add points to the tree that are connectable to the tree, i.e. within
the same connected component, with no restrictions on the connecting
path. If we instead desire the local planner admit \emph{reachable}
points to the tree (which account for some  dynamics), the model
of Fig.~\ref{fig:infillsigmoid} can be tweaked to yield a second
estimate. We also note that considering dynamic models in the
estimation in Eq.~\eqref{eq:pstar} leads to an expansion of the
parameter space in an unmanageable way, so the model estimates
connectivity sans dynamics.

Given $x_1 \in \subscr{C}{free}$, we define the reachable set from $x_1$ as the states
$x_2 \in \subscr{C}{free}$ that can be
connected to $x_1$ by a dynamic, obstacle-free path generated by a
predefined type of control. We choose this type of restricted
reachability\footnote{More generally, a reachable state from $x_1$ is
$x_2$ for which there exists a control $u(t)$ that connects these
states by a dynamic path.~\cite{HLT-AAS-MH:12}}, so we can factor in
system dynamics and remain according to~\cite{LEK-PS-JCL-MHO:96}, who
note that it is preferable to use a very fast local planner even if it
is not too powerful. The oracle marks dynamic paths as not reachable if they are not
obstacle-free, as we are more concerned with testing many solutions
quickly rather than every solution rigorously, even if it may
reachable be with a modified controller or intermediate references.

Next, we provide an upper bound characterization to
reachability from a tree of $M$ nodes in free space by considering the
case where the tree is maximally spread.  This case gives the minimum
effective characteristic length because new samples are connected to
the nearest node, and for a maximally spread tree, the Euclidean distance of
that node to the nearest one in the tree should be the smallest. The
minimum characteristic length maximizes reachability, maximizing the proportion of the database marked as correct,
which enables us to lower bound the number of applications of $Q$ as
per Eq.~\eqref{eq:imaxexact}. An upper bound $p_2^*$
on the average reachability to a set of nodes in random square
lattices is defined in Thm.~\ref{theorem:p2}.

\begin{theorem} \label{theorem:p2}
For a random square lattice $C$ characterized by length $L$, with
concentration $r$ and an arbitrary set $T$ of $M$ nodes in free space,
an upper bound $p_2^*(r,L)$ with $ x_1 \in T \subseteq
\subscr{C}{free}$ and $x_2 \sim \mathbb{U}(C)$ is given by
Eq.~\eqref{eq:pstar} with characteristic length
$L^*=\frac{3L}{\sqrt{M}}$. $p_2^*$ is the absolute upper bound of the
number of correct database solutions, which is related to the number
of times to apply QAA by Eq.~\eqref{eq:imax}.
\end{theorem}

\begin{proof}
The proof follows by considering the best case of a
  maximally spread tree $T$ of $\widetilde{M}$
nodes (and $M$ feasible nodes) within lattice $C$. A tree $T$ with nodes placed according to a
centroidal Voronoi tessellation (CVT)~\cite{QD-VF-MG:99}
of $C$ with $\widetilde{M}$ nodes
  and regions, is one that minimizes the expected distance of every
  node in $C$ to the closest generator. Assume that
  $\widetilde{M}$ is sufficiently large so that there $M$ feasible
  nodes in $\subscr{C}{free}$. A random point will attempt to connect
with the closest parent. Each existing node, when placed according to
a CVT in a convex region, creates a region of connection characterized
(in 2D) by length $\frac{L}{\sqrt{\widetilde{M}}}$
for a $C$ of area $L^2$ with $\widetilde{M}$ regions.  A
CVT, by definition, creates Voronoi regions of connectivity of
expected minimal characteristic length. If a certain node turns out to be infeasible,
the distance of a point to the nearest feasible generator is $\frac{3L}{\sqrt{\widetilde{M}}}$, which can be upper bounded by $\frac{3L}{\sqrt{M }} $. This minimal characteristic length yields a maximum connectivity estimate by substituting $L^*=\frac{3L}{\sqrt{M }} $ for $L$ into Eq.~\eqref{eq:pstar}.
\end{proof}

This is similar to the noted result
in~\cite{LEK-PS-JCL-MHO:96} regarding the restriction of new test nodes to sufficiently
close existing nodes in the tree to maximize the connection
likelihood. In the q-RRT Alg.~\ref{alg:q-RRT}, $p_2^*$ serves to lower bound the number of times QAA must be applied to the database qubit. The intuition behind Thm.~\ref{theorem:p2} is that as the tree grows in number of nodes, it is easier to prove reachability to the tree. The bounding case is when the tree is maximally spread within $C$, as given by a CVT. In that case, the characteristic length can be thought of as $\frac{L}{\sqrt{M}}$, or the original environment size split into $M$ equal sized and roughly convex regions.

\section{q-RRT Results and Discussion}\label{section5}
\subsection{Comparison With Ground Truth} \label{sec:validation-a} 
In the following, we evaluate $p^*$ and $p_2^*$ on a particular
example.
Quantum computers must find the number of correct solutions within the
database using the Quantum Counting
Algorithm~\cite{GB-PH-AT:98}, which is a mix of quantum phase
estimation and Grover's Algorithm. Due to the use of a quantum
computing simulation on a classical computer, this value is knowable.
To ascertain reachability, the oracle $\mathcal{X}$ uses the following
robot dynamics and control law and performs reference tracking from an
$\subscr{x}{parent}$ to a $\subscr{x}{new}$,
\begin{equation*} 
  x(t+1) = 
  Ax(t)+Bu(t)\:,
  \: x(0) = \subscr{x}{parent},
\end{equation*}
\begin{equation*}
  A = \begin{bmatrix}-1.5&-2\\
    1&3\end{bmatrix}, \,
  B= \begin{bmatrix}0.5&0.25\\0&1\end{bmatrix},
\end{equation*}
\begin{equation*}
    u(t) = -Kx(t),
\end{equation*}
\begin{equation*}
   K = \begin{bmatrix}1.9&-7.5\\
    1&7\end{bmatrix}.
\end{equation*}
The constant gain matrix $K$ can be any matrix such that the closed
loop system is stable.

In Fig.~\ref{fig:pstarest}, we compare $p^*$ and $p_2^*$ against a
histogram of $250$ simulations, in randomized environments, of a $2^{11}$ size database. Each of the $250$ simulations are grouped according to the
proportion of correct database solutions they provide. Cases are run
with $L = 32$, $r = 0.6$, and with a tree of $M = 5$ nodes. From the
figure, observe that $p^*$ forms a slightly high estimation, while
$p_2^*$ is validated to form an upper bound.
Since $p^*$ refers to mean connectivity, and
not reachability, the proportion of correct solutions, when we factor in provable reachability, is generally less than $p^*$. This explains why $p^*$ forms a slightly high estimation of the mean.
On the other hand, $p_2^*$
correctly upper bounds the proportion in the case where $M=5$.

\begin{figure}[h]
	\centering
	\includegraphics[width=.5\textwidth]{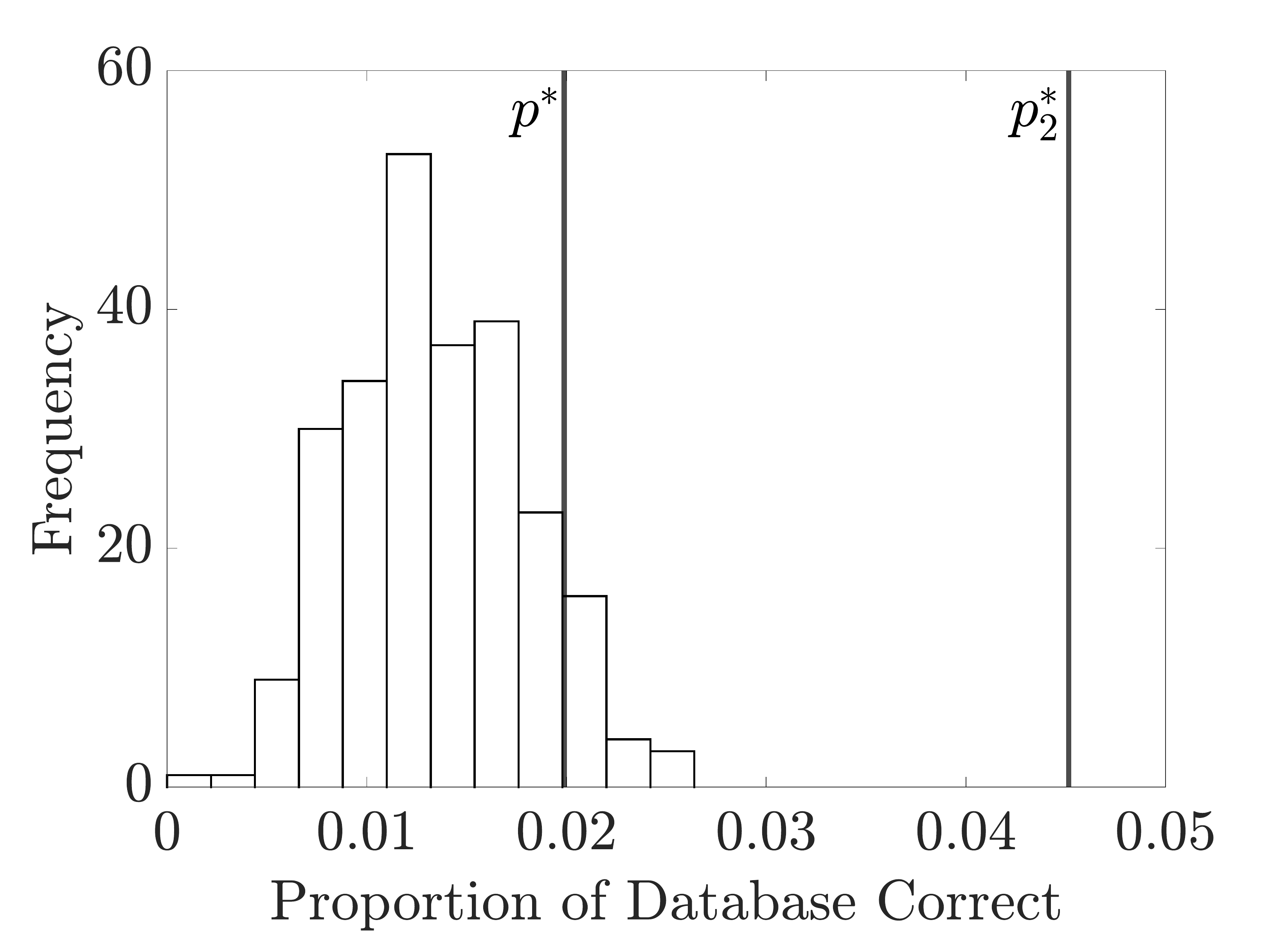}
	\caption{Comparison of $p^*$ and $p_2^*$ with a histogram of
          $250$ random cases of database size $2^{11}$ with $L = 32$,
          $r = 0.6$, and a tree of $M = 5$ nodes.}
	\label{fig:pstarest}
\end{figure}

\subsection{Results in Dense Random Lattices}
In this section, we show the results of q-RRT creating a tree within
large connected components of random 2D lattices (explained in Section~\ref{section4}), as shown in Fig.~\ref{fig:environment}. We compare algorithm
performance with a classical, and largely identical, version of RRT
attempting to span the same connected component. The classical version
of RRT replaces the quantum database search with a classical oracle
check on a single point. All path planning simulations are performed in a $2$D environment run
with Matlab v2022b on a desktop computer with an Intel i5-4690K CPU
and an AMD RX 6600XT GPU. A selection of Matlab code is available at github.com/pdlathrop/QRRT. The quantum states and algorithms are
simulated using the Matlab Quantum Computing Functions
library~\cite{CF:03}. Simulations are run in a random square lattice
of size $L = 72$, where each method is
given a random start node within the largest connected
component. In the simulations, $r$ varies because concentration creates large differences in performance of both algorithms.
Fig.~\ref{fig:resultOracle} and Fig.~\ref{fig:resultTime}
show the average number of oracle calls and average real run-time of
each algorithm to create an $11$ node tree, which is an arbitrary number chosen to showcase average performance. Each data point is averaged over $50$
planning problems, in $50$ random environments. Both algorithms are tested against the same environments.

\begin{figure}[h]
	\centering
	\includegraphics[width=.5\textwidth]{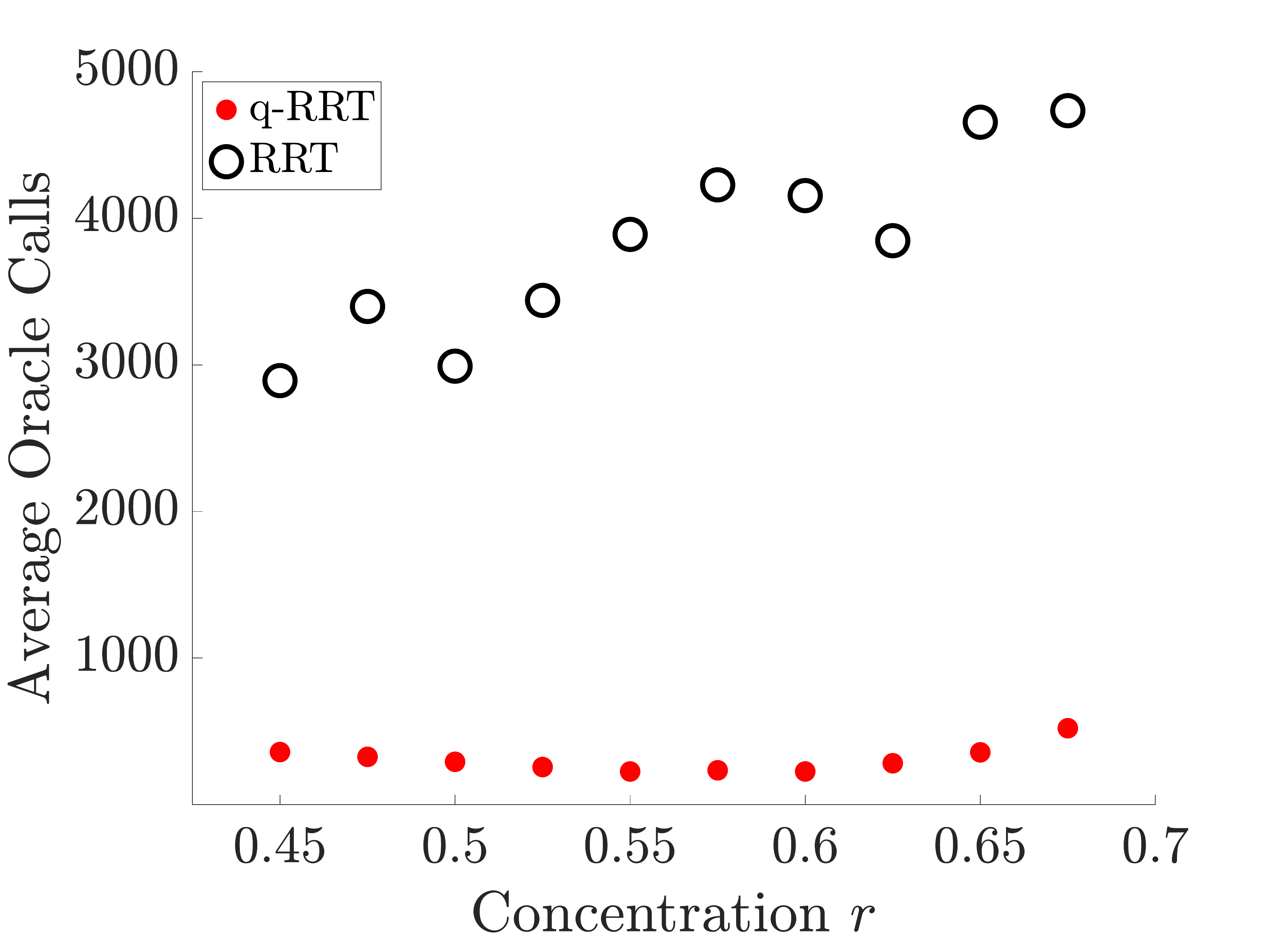}
	\caption{Comparison of the average number of oracle calls by q-RRT and RRT as concentration varies, for $L = 72$.}
	\label{fig:resultOracle}
\end{figure}

\begin{figure}[h]
	\centering
	\includegraphics[width=.5\textwidth]{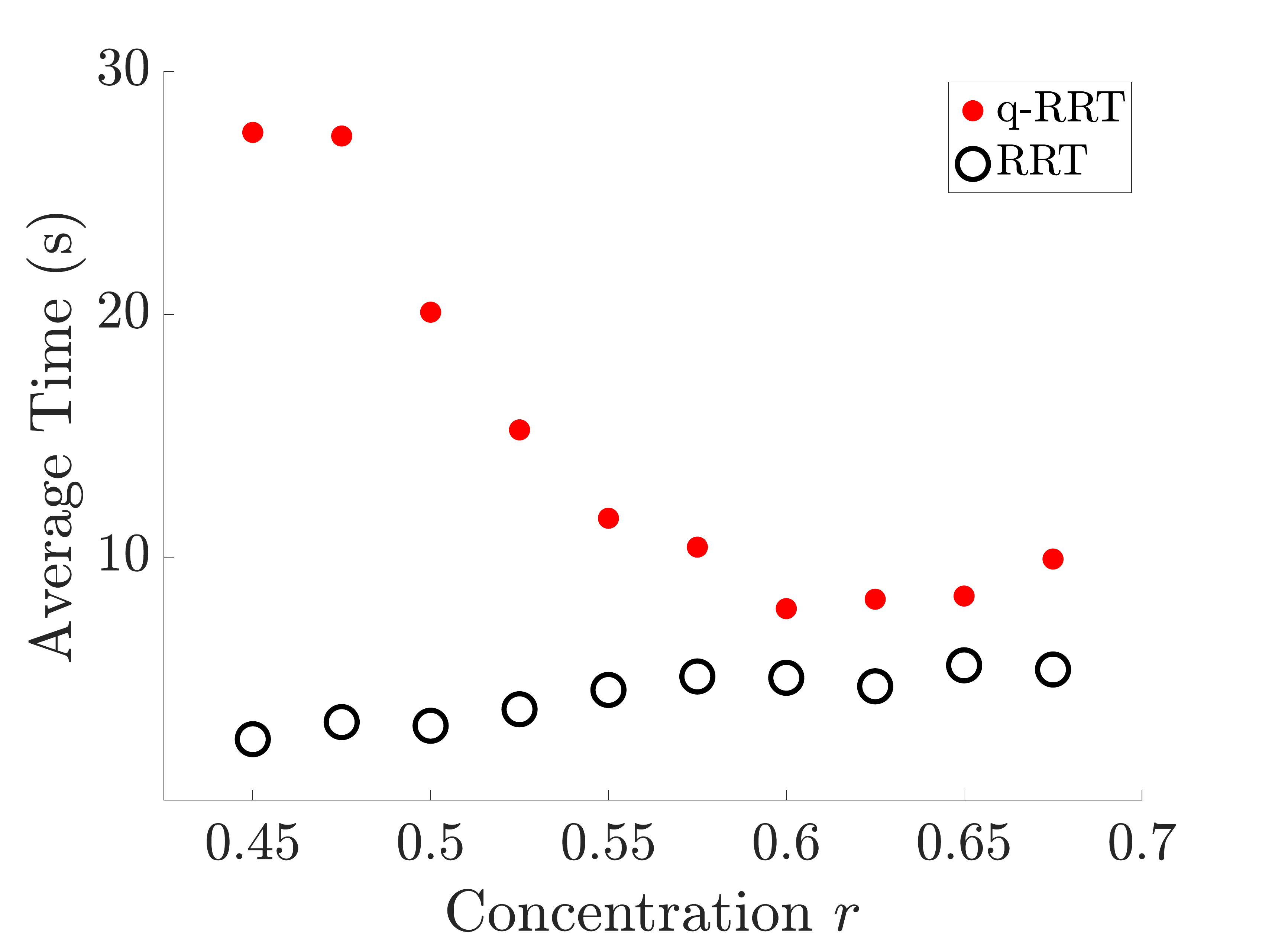}
	\caption{Comparison of the average real run-time of q-RRT and RRT as concentration varies, for $L = 72$.}
	\label{fig:resultTime}
\end{figure}

Over the range $[0.45,0.7]$ in concentration $r$, q-RRT averaged $308$
oracle calls, while the classical RRT averaged $3820$ oracle
calls, as a result of a quadratic performance increase. Algorithm q-RRT averaged $14.7$ seconds per case, compared to RRT's
average of $4.3$ seconds, and this is due to the implementation of quantum algorithms via arrays on a classical computer. On a quantum computer, the actual run-time
advantage would be proportional to the average oracle call
advantage. As $r$ increases, the average number of oracle calls also
generally increases due to the increased difficulty in making
connections in denser environments. For RRT, the average time per case
shows this increase because most of the algorithm run-time is in
performing reachability tests. For q-RRT, as $r$
increases, there is an initial decrease in average run-time, possibly
because at low $r$, the largest connected component tends to be very non-convex
and widely spread. This causes additional reachability tests to
be performed because of a run-time
optimization where points are excluded first based on whether they are
not in the same connected component, then based on
reachability.
 
Minimizing oracle calls is useful in situations where admitting points
to the tree has high computational cost, or where reachability checks
carry a cost. In our method, the oracle tests experimentally whether a
possible point (or database of points) is within the reachable set,
which is a complex problem to solve analytically for non-simple
systems~\cite{EA-TD-GF-AG-CLG-OM:06}. This can result in
significant time savings, and in some cases may allow offline
algorithms to become online. In large dense environments where most random points are not
admittable to the tree, many reachability tests must be performed
to admit even a single valid state. In such situations,
q-RRT far outperforms RRT in the ability to admit new nodes to the tree (per oracle call).
\subsection{Database Size Comparison}
In this section, we show the effect of variance of total database
size $2^n$ on the performance of q-RRT as compared to classical
RRT. In Fig.~\ref{fig:databasesizeoracle} and
Fig.~\ref{fig:databasesizetime} we show the average number of oracle
calls and the average real run-time, respectively, of q-RRT with
databases sized $2^8$ and $2^9$ for $L = 72$ while concentration
varies. Again, we compared q-RRT with RRT in creating a tree with $11$ nodes, and each data point is averaged over $50$
planning problems, in $50$ random environments.
\begin{figure}[h]
	\centering
	\includegraphics[width=.5\textwidth]{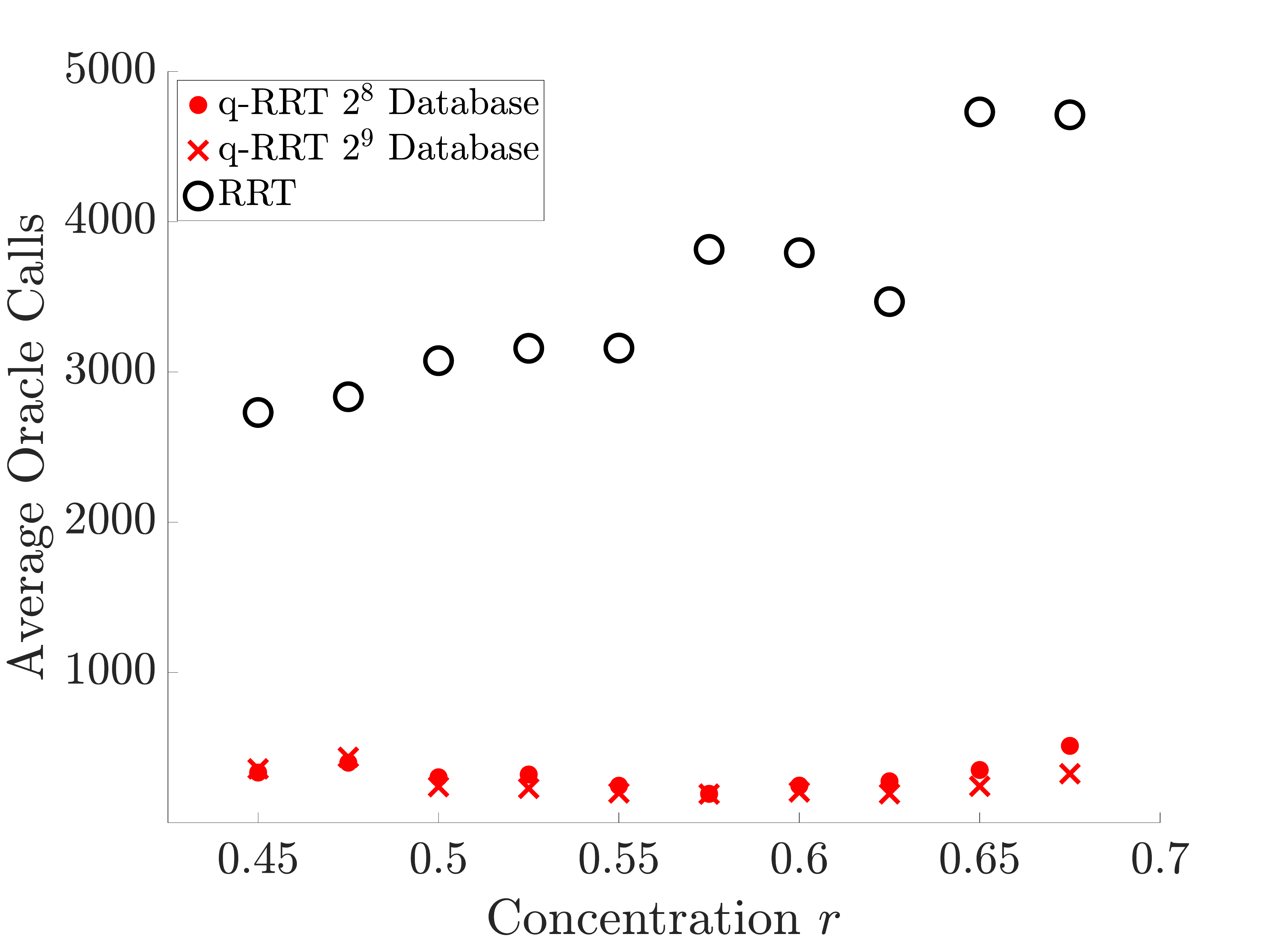}
	\caption{Comparison of the average number of oracle calls by
          q-RRT and RRT as concentration varies, for $L = 72$ and
          Database sizes $2^8$ and $2^9$.}
	\label{fig:databasesizeoracle}
\end{figure}
\begin{figure}[h]
	\centering
	\includegraphics[width=.5\textwidth]{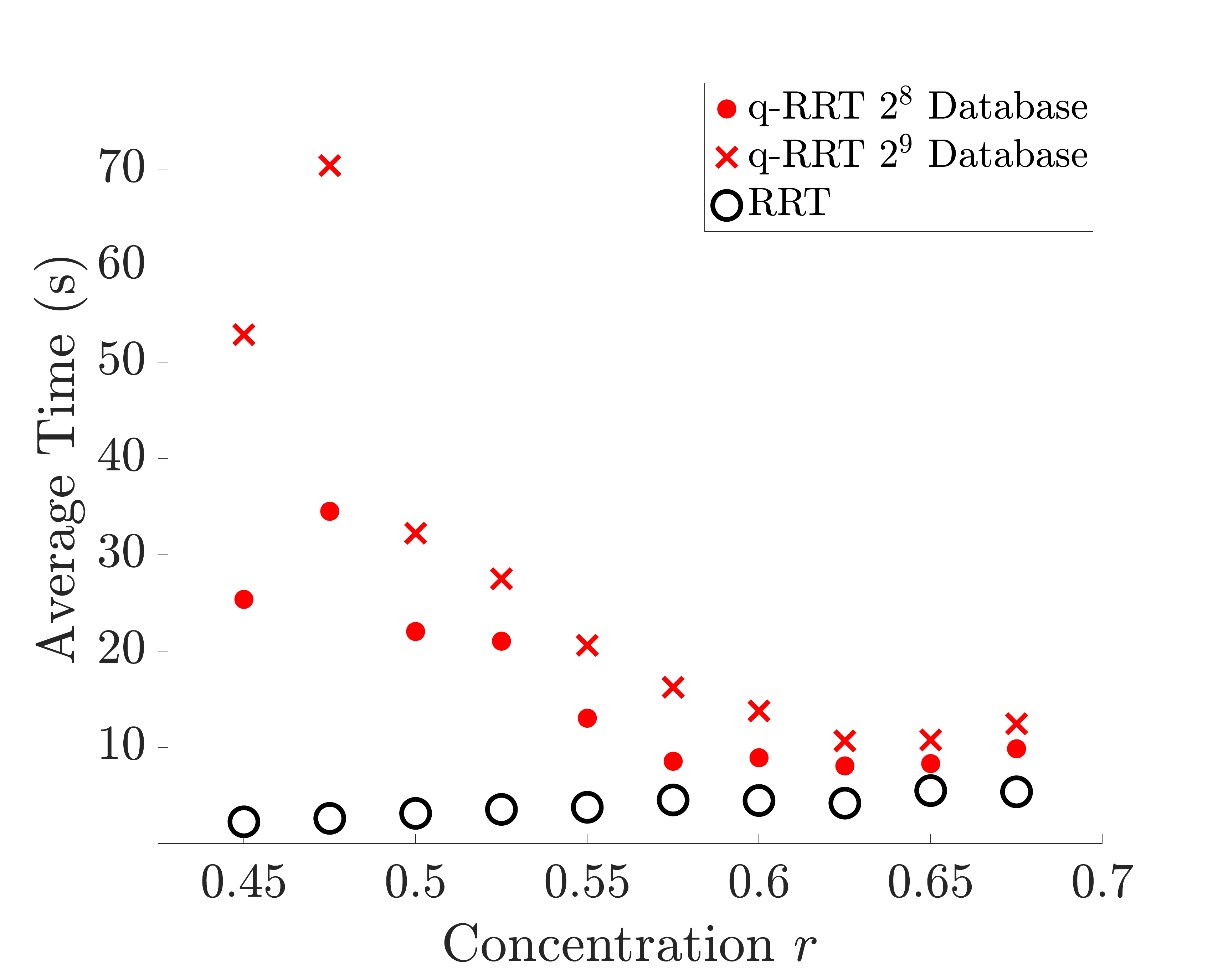}
	\caption{Comparison of the average run-time of q-RRT and
          RRT as $r$ varies, for $L = 72$ and Database sizes
          $2^8$ and $2^9$.}
	\label{fig:databasesizetime}
\end{figure}

We verify that changes to database sizing does not effect the overall
trend of average number of oracle calls or average run-time over
varying concentration. The larger $2^9$ sized database
resulted in lowered average oracle calls, especially at higher $r$,
when compared to $2^8$ sizing. This is consistent with the main reason
the quantum algorithm provides a reduction at all, which is the
ability to perform reachability tests on many possible states
simultaneously. Predictably, with respect to real run-time on a
classical computer, larger database versions of q-RPM take longer
across all $r$, as more reachability tests need to be performed with
the quantum computing simulation. However, on a quantum device, we
expect the run-time to be analogous to the number of oracle calls.

\subsection{Oracle Call Constraint}
In this section, we identify an approach for tree construction that limits the
optimal number of oracle calls to a maximum of $N_{\mathcal{X}}$
per node added to the tree. 
We may want to create databases
of correctness proportion $p$, rather than just predict $p$ from the
environmental parameters, especially in time limited cases. In order
to limit the number of optimal oracle calls to $N_\mathcal{X}$, we
constrain the $L_1$ (Manhattan) distance to evaluate reachability to be equal to a
certain value, which we call $D_{L1}$. This will ensure that the number of successful
reachable connections becomes higher in cluttered environments, thus
requiring a smaller number of oracle calls.
We consider an environment with
a fixed $L$ value and measure distance in terms of the $1$-norm or
Manhattan distance. The $1$-norm is chosen over the Euclidean distance as,
intuitively, it can yield a superior parameter for estimating
connectivity within a square lattice. In this context, the word optimum refers to the
number of oracle calls that maximizes the likelihood of measuring a
correct database element.
 
\begin{figure}[h]
	\centering
	\includegraphics[width=.5\textwidth]{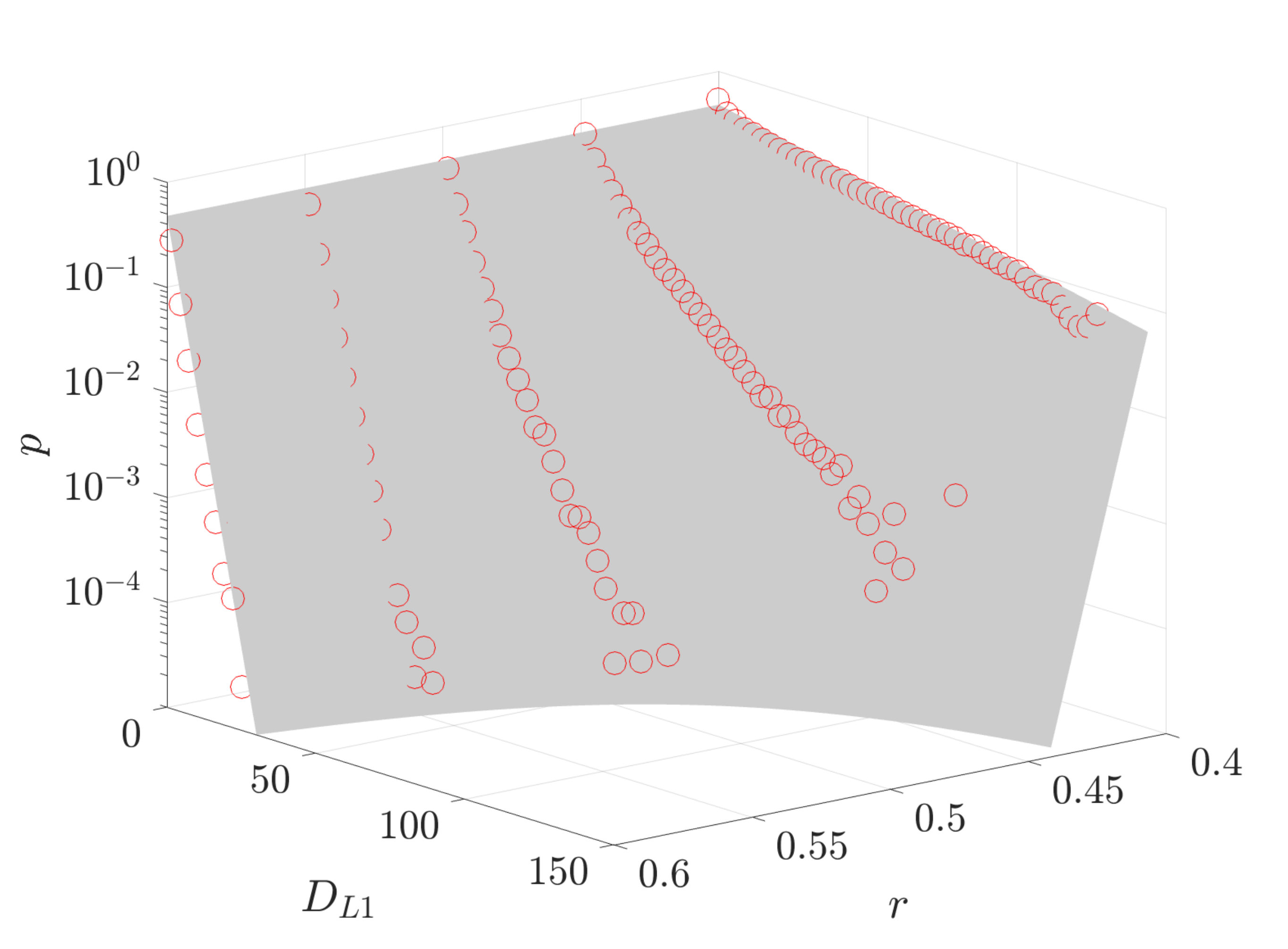}
	\caption{Semilog plot of numerically generated data points
          ($\circ$) estimating $p$, the likelihood of free-random
          point connectivity as a function of $D_{L1}$, the L$_1$
          distance between parent and child, for various
          concentrations $r$. Eq.~\eqref{eq:L1} is depicted as the
          gray surface. Data is generated with $L = 72$.}
	\label{fig:L1Connectivity}
\end{figure}

In Fig.~\ref{fig:L1Connectivity}, we show how
average connectivity $p$ scales according to a negative exponential
with increasing L$_1$ distance between parent and child. Values spread at the larger L$_1$ distances due to smaller sample sizes. For a given concentration
$r$ and an oracle call constraint $N_\mathcal{X}$,
Fig.~\ref{fig:L1Connectivity} can be used to select the maximum
$D_{L_1}$ that will select $N_\mathcal{X}$ as the approximate number of optimal oracle calls. To exemplify how such a tool can be used, we fit a model using nonlinear least squares to
the numerical $L = 72$ data shown in Fig.~\ref{fig:L1Connectivity}, which takes the form,
\begin{equation}\label{eq:L1}
    p = a\;e^{(br+c)\;D_{L_1}},
\end{equation}
with $a = 0.479$, $b = -1.72$, and $c = 0.674$ with coefficient of
determination $R^2 = 0.981$. Again, over-fitting is not a concern for
three parameters modeling $336$ data points.  Equivalently,
\begin{equation}\label{eq:L1inverse}
    D_{L_1} = \text{ln}\left(\frac{\pi^2}{16N_\mathcal{X}^2a}\right)/(br+c),
\end{equation}
when Eq.~\eqref{eq:L1} is solved for $D_{L_1}$ and $p$ is
related to $N_\mathcal{X}$ via Eq.~\eqref{eq:imax}.
Eq.~\eqref{eq:L1inverse} allows an algorithmic distance
constraint to be found from an oracle call constraint and the
environment concentration. From here, when q-RRT is building a
database, states should be randomly selected from the boundary of a
ball at radius $D_{L_1}$.  In order to further restrict oracle calls,
q-RRT can instead randomly sample within a ball of radius $D_{L_1}$,
which results in a lower mean $L_1$ distance, and therefore higher $p$
and lower number of oracle calls. 

\begin{figure}[h]
	\centering
	\includegraphics[width=.5\textwidth]{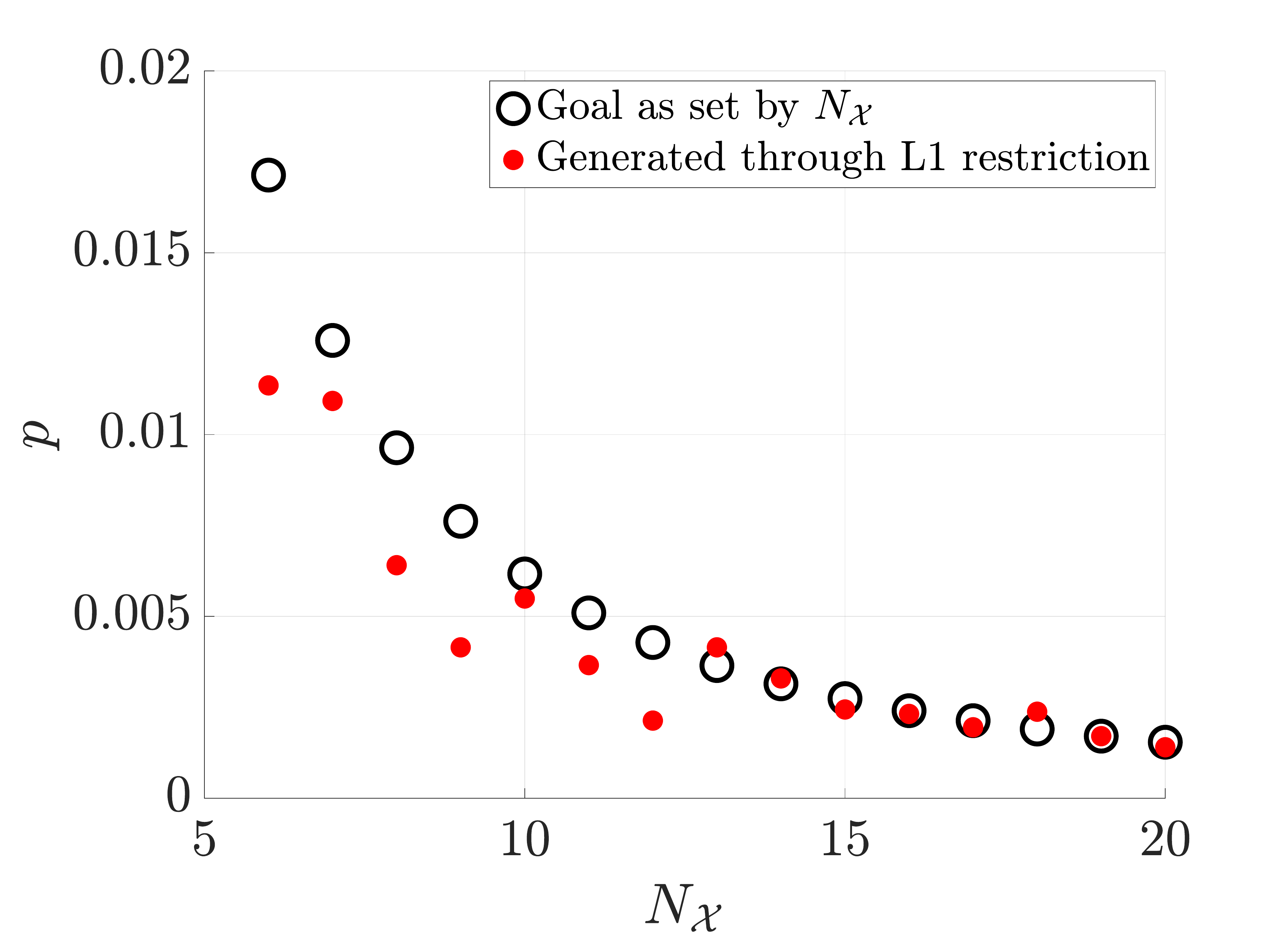}
	\caption{Evaluation of the ability to select $p$ given
          $N_\mathcal{X}$ using Eq.~\eqref{eq:L1}. Data
          is generated with $L = 72$, $r = 0.5$, and the database size
          $2^{14}$.}
	\label{fig:L1proof}
\end{figure}
An analysis of the ability to select $p$ given a
concentration $r$ and oracle call constraint $N_\mathcal{X}$ is given
in Fig.~\ref{fig:L1proof}. Datasets of size $2^{14}$ are constructed
in a random square lattice of $r = 0.5$ and $L = 72$ (chosen for discretizability) for various
$N_\mathcal{X}$. The goal points are the
$p$ that correspond with an optimum number of oracle calls
$N_\mathcal{X}$ to admit one node to the RRT. Therefore, the number of
oracle calls which yields the maximum likelihood of adding $M$ nodes
to the RRT from $M$ database creations is $MN_\mathcal{X}$. The use of
an $L_1$ restricting version of q-RRT alongside
Eq.~\eqref{eq:L1inverse} allows the creation of an $M$ node RRT
where $N_\mathcal{X}$ has been approximately chosen as the optimum number of oracle
calls per node. Fig.~\ref{fig:L1proof} shows that, as $N_\mathcal{X}$ increases, we have a more accurate ability  to select $p$.

\section{Conclusion}
The goal of this work was to provide a first study of the application of quantum algorithms to sampling based robotic motion planning. 
We developed a full path quantum search algorithm for sparse environments
and a Quantum RRT algorithm for dense random square lattices. The
q-RRT algorithm uses Quantum Amplitude Amplification to search a
database of possible parent-child relationships for reachable states
to add to the tree. q-RRT, tested on a simulated quantum device, successfully employs a quadratic speedup of
database searches to reduce oracle reachability calls when
constructing a tree. We also provide key quantum measurement
probability results, and tools for estimating and selecting the number
of correct database entries using numerical modeling and guided
sampling. Future work includes path planning employing quantum mean
estimation for uncertainty modeling~\cite{PS:21}, implementing q-RRT
on disjoint trees~\cite{TL-FR-GF:19} for planning over multiple
disconnected components, employing a parallel quantum computing
structure to q-RRT, and exploring path-optimality based algorithms in the context of quantum computing.


\bibliographystyle{IEEEtran}
\bibliography{alias,SM,SMD-add,JC}

\begin{IEEEbiography}[{\includegraphics[width=1in,height=1.25in,clip,keepaspectratio]{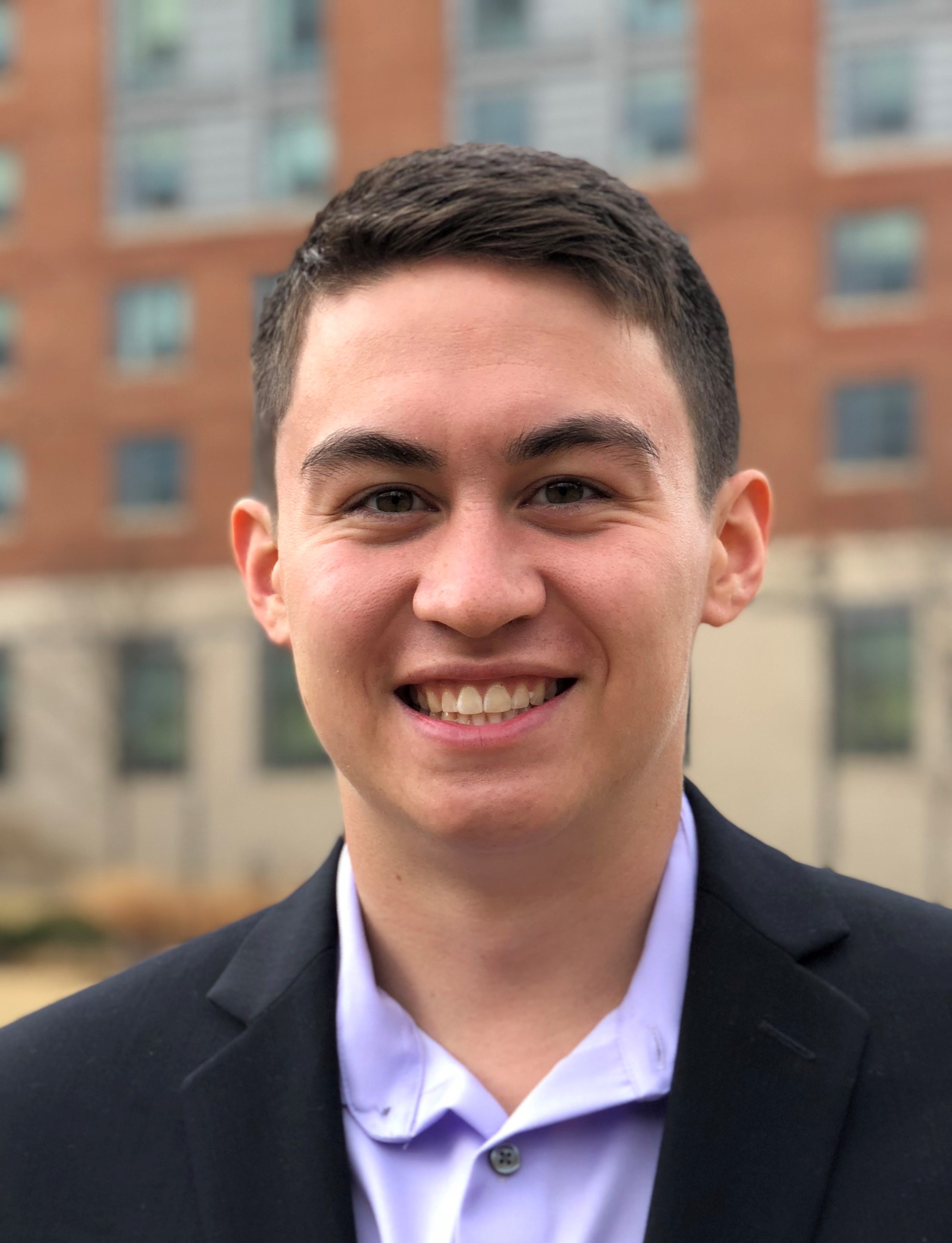}}]{Paul Lathrop} received the B.S. degree in Aerospace Engineering from the University of Maryland, College Park in 2019, the M.S degree in Aerospace Engineering from the University of California, San Diego in 2021, and is currently a Ph.D. candidate in Mechanical and Aerospace Engineering at the University of California, San Diego. His research interests include robotic motion planning and quantum algorithms.
\end{IEEEbiography}
\begin{IEEEbiography}[{\includegraphics[width=1in,height=1.25in,clip,keepaspectratio]{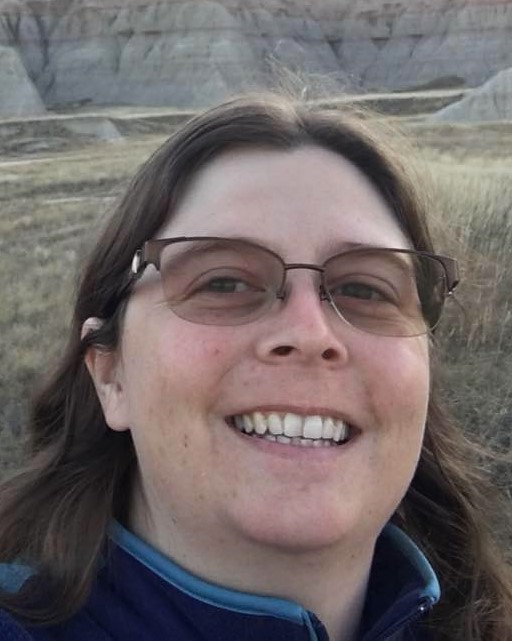}}]{Beth Boardman}
received a Ph.D. degree in aerospace engineering from the University of California, San Diego, USA in 2017 and a B.S. and M.S. in aeronautics and astronautics from the University of Washington, Seattle, USA in 2010 and 2012, respectively. She has worked as a Research and Development Engineer at Los Alamos National Laboratory since 2018. Her research interests include robotics and automation. Beth is currently the Team Leader for the Process Control and Robotics team in the Process Automation and Control group.
\end{IEEEbiography}
\begin{IEEEbiography}[{\includegraphics[width=1in,height=1.25in,clip,keepaspectratio]{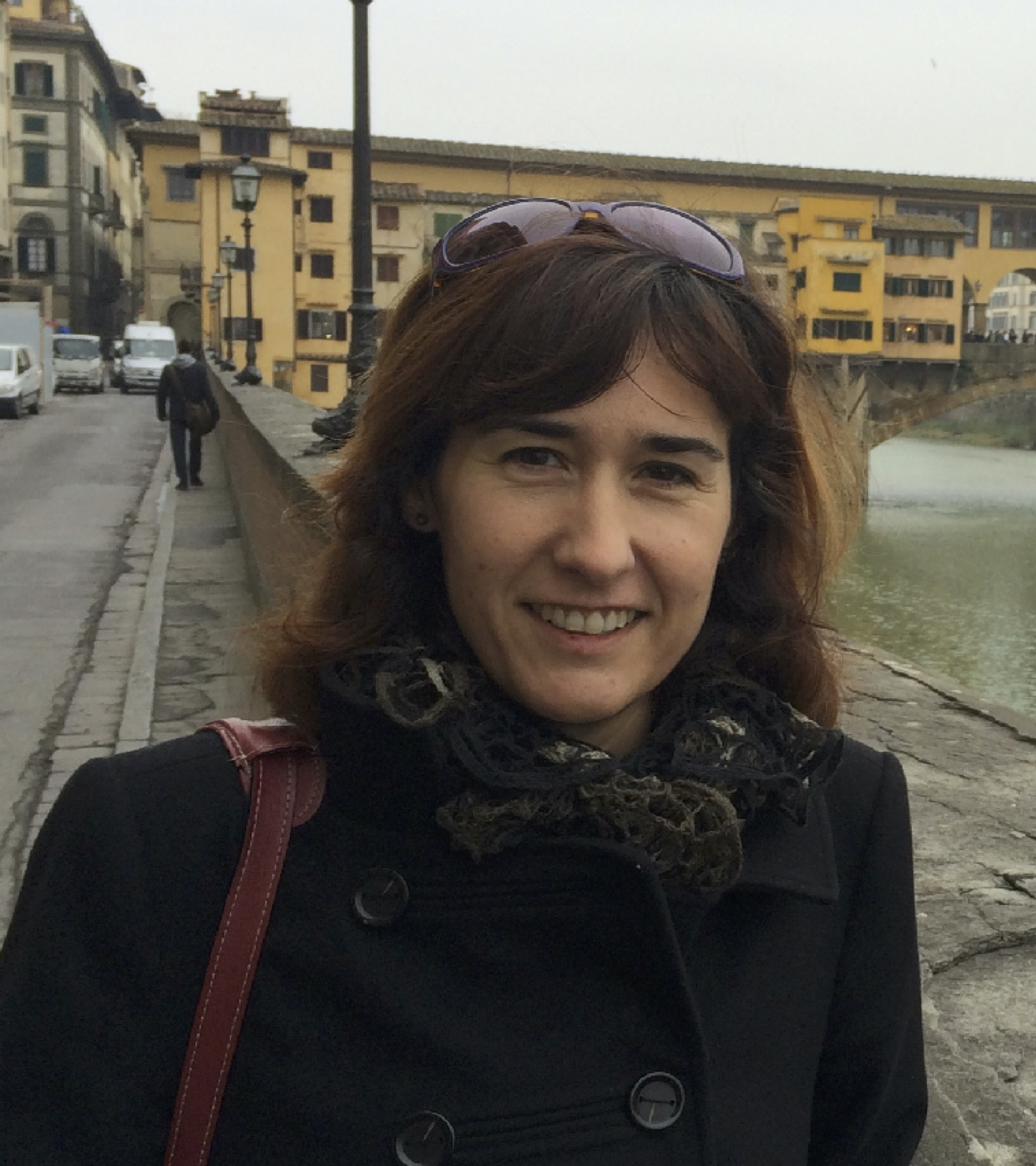}}]{Sonia Mart\'{i}nez}
 (M'02-SM'07-F'18) is a Professor of Mechanical and
  Aerospace Engineering at the University of California, San Diego,
  CA, USA. She received her Ph.D. degree in Engineering Mathematics
  from the Universidad Carlos III de Madrid, Spain, in May 2002. She
  was a Visiting Assistant Professor of Applied Mathematics at the
  Technical University of Catalonia, Spain (2002-2003), a Postdoctoral
  Fulbright Fellow at the Coordinated Science Laboratory of the
  University of Illinois, Urbana-Champaign (2003-2004) and the Center
  for Control, Dynamical systems and Computation of the University of
  California, Santa Barbara (2004-2005).  Her research interests
  include the control of networked systems, multi-agent systems,
  nonlinear control theory, and planning algorithms in robotics. She
  is a Fellow of IEEE. She is a co-author (together with F. Bullo and
  J. Cort\'es) of ``Distributed Control of Robotic Networks''
  (Princeton University Press, 2009). She is a co-author (together
  with M. Zhu) of ``Distributed Optimization-based Control of
  Multi-agent Networks in Complex Environments'' (Springer, 2015).
  She is the Editor in Chief of the recently launched \textit{CSS IEEE Open
  Journal of Control Systems.}
\end{IEEEbiography}

\EOD
\end{document}